\documentclass[hidelinks]{amsart}
\usepackage{amsthm}
\usepackage{tikz}
\usepackage{tikz-cd}
\usetikzlibrary{matrix,arrows,decorations.pathmorphing}
\usepackage{amsmath}
\usepackage{amsfonts}
\usepackage{verbatim}
\theoremstyle{plain}
\newtheorem{lem}{Lemma}[section]
\newtheorem{theorem}[lem]{Theorem}
\newtheorem{cor}[lem]{Corollary}
\newtheorem{deff}[lem]{Definition}
\newtheorem{prop}[lem]{Proposition}
\newtheorem{rem}[lem]{Remark}
\usepackage{graphicx}
\usepackage{hyperref}
\newcommand{\proj}{\ensuremath{\hat{\mathfrak{R}}}}
\newcommand{\ef}{\mathfrak{F}}
\numberwithin{equation}{section}

\usepackage{fullpage}
\usepackage{color}
\usepackage{leftidx}

\numberwithin{equation}{section}
\title{Periods of Feynman diagrams and GKZ D-modules}
\author{Emad Nasrollahpoursamami} 
\begin{document}
\maketitle
\begin{abstract}
We study differential equations for Feynman amplitudes and we show that the corresponding D-module is isomorphic to a GKZ D-modules. We show that the sheaf of solutions to the D-module is isomorphic to a
certain relative homology and the amplitudes are periods of a relative motive. Using these ideas, we develop a method of regularization which specializes to dimensional regularization and analytic regularization.
\end{abstract}
\tableofcontents
\section{\textbf{Introduction}}

In perturbative quantum field theory the scattering amplitudes, which are the probabilities of physical processes, can be approximated by sums over Feynman diagrams. Feynman diagrams are graphs corresponding to certain integrals. The integral corresponding to each Feynman diagram is a function of parameters called external momenta. The resulting functions are called {\em amplitudes} of Feynman diagrams. To compute the actual scattering amplitudes one needs to add these functions, hence understanding the properties of these functions is necessary for both experimental and theoretical physics. In this paper, we restrict our attention to the case of a scalar field theory. This means that the external momenta are just vectors in $\mathbb{R}^D$, with $D$ the 
dimension of the theory. Suitable generalizations exist for arbitrary quantum field theories. \\

Amplitudes of Feynman diagrams are functions of external momenta. Amplitudes are functions on $\mathbb{R}^{D(|V|-1)}$, with $D$ the dimension of the theory and $V$ the set of vertices of the diagram. We map this vector space to another vector space related to the graph, denoted by $\mathfrak{V}_{\Gamma}$ so that the amplitude is the pull back of a (multi-valued)function on $\mathfrak{V}_{\Gamma}$. On this new vector space, we construct a holonomic regular $D$-module, of which the function we are considering is a solution. As a result, we show that the Feynman amplitude satisfies a holonomic regular system of differential equations.\\

The differential equations on $\mathfrak{V}_{\Gamma}$ are a special case of GKZ or A-Hypergeometric system of differential equations introduced in \cite{GKZH} and \cite{GKZ}. It follows from the result in these references that the corresponding D-module is holonomic and regular. It is well known that these D-modules come from twisted Gauss-Manin connections on toric varieties. Recently it is shown in \cite{zhu} that, in the Calabi-Yau case, the relative homology computes the sheaf of solutions. Using results of \cite{andre} we show that their construction can be generalized to a non-Calabi-Yau case which includes 
Feynman diagrams.\\

In \cite{BEK} the authors show that the Feynman amplitude for primitive log divergent graphs is a period of the complement of a hypersurface in a toric variety relative to the boundary. In their work the toric variety is constructed using iterated blow ups so that it separates graph hypersurface from integration cycle. Using our result we directly construct a toric variety with such a property without any conditions on the diagram. We show that the variation of the hypersurface corresponds to the differential equations. Since the construction is explicit we can compute the cohomology using generators and relations. The dimension of the relative cohomology is the normalized  volume of a polytope closely related to a matroid polytope. The matroid here is the graphical matroid corresponding to the diagram. \\

The explanation above was for the convergent case. In the physics literature, dimensional regularization is a formal way of getting finite numbers from divergent integrals. The idea is to formally find the Laurent expansion of the integral with respect to $D$. In \cite{brosnan} and \cite{BW} the authors show that the coefficients of the expansion are periods, but their results are not explicit. Using the description of cohomology with generators and relations, we show how one can define the integral for the divergent case. In particular, we prove the following theorem. 
\begin{theorem}\label{mm}
Given a graph $\Gamma$ with n edges and first Symanzik polynomial $\Psi$ and second Symanzik polynomial $Q$ (including mass terms), the amplitude in dimensional regularization, up to a constant, can be computed by the following integral:
$$c_0\mathcal{A}(D/2+\epsilon)=\int_{\mathbb{R}_+^n} e^{-Q/\Psi} \frac{1}{\Psi^{D/2+\epsilon}}=\sum_{i\geq-n} \epsilon^i A_i(D/2)$$
The left hand side is meromorphic and poles can be described in the following way. For a 2-connected subgraph $\gamma \subset \Gamma$, let $\ell_{\gamma}$ be the dimension of the first homology of $\gamma$. $|E(\gamma)|$ is the number of edges of $\gamma$. $\mathcal{A}(D/2)$ has a pole at $D/2 \in \mathbb{C}$ iff 
$$D/2\ \ell_{\gamma} - |E(\gamma)| \in \mathbb{Z}_{\geq 0}$$
for a 2-connected subgraph $\gamma$. Furthermore, the lowest coefficient, $A_{-n}(D/2)$, at integers comes from a pairing between an algebraic relative cohomology class and a Betti homology class explicitly constructed in section \ref{GKZsection}. 
\end{theorem}

In Section \ref{feynmann} we define the Feynman amplitude for a Feynman diagram and show how one can present it in the parametric form. The new result is that a product of a power of the first Symanzik polynomial and the second Symanzik polynomial is the determinant of a matrix. We also show that the coefficients of the first and second Symanzik polynomials are norms of Pl\"ucker coordinates for a Grassmannian naturally defined by the graph.\\

In section \ref{sec1} we study how the integral changes as we vary the coefficients of the first and second Symanzik polynomials. In the convergent case, we find a set of linear PDEs satisfied by the integral. Using analytic continuation we define the integral for the divergent case but the proof is not constructive. We first show that the analytic continuation exists and, using that, we find the PDEs satisfied by coefficients in the Laurent expansion. This method is based on \cite{BW}. It turns out that the set of differential equations is a special case of the GKZ differential equations for the convergent case. For the divergent case, the coefficients of the Laurent expansion are solutions to iterated extension of GKZ differential equations. \\

In section \ref{GKZsection} we study the GKZ differential equations as a D-module on a vector space $V$. Given a polynomial $f$ in $n$ variables with Newton polytope $A$ such that the dimension of $V$ is the number of points in $A$, and a vector $\beta \in k^{n+1}$, we consider the corresponding GKZ D-module $H_{1 \times A}(\beta)$. We construct a projective toric variety $\mathbb{P}_{\Sigma}$ together with a line bundle on it. The vector space of global sections of the line bundle is isomorphic to $V$. Let $D$ be the complement of the torus $\mathbb{T}$ in $\mathbb{P}_{\Sigma}$. Let $U$ be the complement of the zeros of $f$ in $V\times \mathbb{P}_{\Sigma}$, where $V$ parametrizes the coefficients of the polynomial $f$. Given $v\in V$ we show that the algebraic relative cohomology of the pair $(U_v,D\cap U_v)$ with the Gauss-Manin connection is isomorphic to $H_A(\beta)$ as D-module, where $U_v$ is the fiber of $U$ over $v$. Using the Riemann-Hilbert correspondence, we deduce that the cycle to period map gives us a complete set of solutions.\\

In section \ref{integralsection} we study integrals over positive real points of the toric variety, considered as a  chain in relative homology. We show that a necessary and sufficient condition for an integral of the type
$$
\int_{\mathbb{R}_{+}^{n}} \frac{f^{\beta_0}}{t_1^{\beta_1} \dots t_n^{\beta_n}}dt_1\dots dt_n
$$
to converge, is $\beta$ being semi non-resonant, defined in Definition \ref{seminon}. Using relations in the cohomology ring, we develop a method to define this integral for any value of $\beta$ by meromorphic continuation. We show that the poles of this function appear in translates of the faces of a cone in $\mathbb{R}^{n+1}$. This cone is the cone over the Newton polytope of $f$. \\

In section \ref{feynGKZ} we apply the methods developed in the previous sections to amplitudes and we explicitly construct a motive such that its periods give us the amplitude. We completely describe the Newton polytope in this case and we show that its facets correspond to the product of subgraphs and quotient graphs. Using regularization methods, we define the $\epsilon$ expansion of the amplitude for divergent graphs. We do not use resolution of singularities and the construction is explicit. 

\section{\textbf{Feynman Diagrams}}\label{feynmann}

Feynman diagrams (or Feynman graphs) are one dimensional simplicial complexes with half edges attached to some of the vertices. These half edges are called {\em external edges}, while all other one dimensional simplices are called {\em internal edges}. In the physics literature, for each external edge, it is common 
to fix a vector in $\mathbb{R}^D$.  These vectors are called {\em external momenta}. 
They are subject to a momentum conservation law, given by the requirement that the sum of all 
external momenta of the graph is zero.  Since the amplitude only depends on the sum of the 
external momenta at each vertex, we can equivalently assign a momentum vector to each vertex 
of the graph and forget about external edges. Namely, we assign to a vertex the sum of 
the external momenta of all the external edges attached to 
that vertex, or zero if there are no external edges at that vertex. 
So in the following external momenta will always be assigned to vertices.\\

Let $E$ be the set of edges of the graph and let $V$ be the set of vertices. 
We have an exact sequence of free $\mathbb{Z}$ modules 
\begin{equation}\label{exseq1}	
 0\rightarrow H_1(\Gamma,\mathbb{Z})\xrightarrow{\hspace{5mm}\eta'\hspace{5mm}} \mathbb{Z}^{|E|} \xrightarrow {\hspace{12mm}}\mathbb{Z}^{|V|-1}\rightarrow 0 , 
\end{equation}  
where $H_1(\Gamma,\mathbb{Z})$ is the first homology of the graph $\Gamma$ with coefficients 
in $\mathbb{Z}$, that is, the free $\mathbb{Z}$ module generated by loops. 
The morphism on the right is the boundary map. Note that, to define this map, 
we need to fix an orientation on the edges of the graph, but the final result is independent 
of this choice. Taking the tensor product of the  sequence above with $\mathbb{R}^D$ gives the  
exact sequence
\begin{equation}\label{exseq2}
 0\rightarrow H_1(\Gamma,\mathbb{R}^D)\xrightarrow{\hspace{5mm} \eta \hspace{5mm} } \mathbb{R}^{D|E|} \xrightarrow{\hspace{5mm} \beta \hspace{5mm}} \mathbb{R}^{D(|V|-1)}\rightarrow 0 .
\end{equation}
Note that the choice of external momenta 
$\{p_v\in \mathbb{R}^D |v\in V,\sum_{v\in V}p_v=0\}$ is just a choice of a 
vector $a$ in $\mathbb{R}^{D(|V|-1)}$.
\\

Define $Q_e=P_e^2+m_e^2:\mathbb{R}^{|E|D}\rightarrow \mathbb{R}$, where $P_e^2$ is given by
first projecting onto the $D$ coordinates corresponding to $e$ and then taking the sum 
of the squares of these $D$ coordinates. 
\\

Let $\tilde{a} \in \mathbb{R}^{D|E|}$ be a lift of $a$, under the map $\beta$ of \hyperref[exseq2]{\eqref{exseq2}}.

\begin{deff}\label{Amplitude}
The amplitude of a Feynman graph $\Gamma$ with external momenta $a \in \mathbb{R}^{D(|V|-1)}$
is given by the integral
\begin{align}\label{Ampl}
\begin{split}
\mathcal{A}(\Gamma,a,m_e):&=\int_{\beta^{-1}(a)} \prod_e \frac{1}{Q_e} (\eta+\tilde{a})_*(d\mu)\\
&=\int_{H_1(\Gamma,\mathbb{R}^D)} (\eta+\tilde{a})^*(\prod_e \frac{1}{Q_e}) d\mu 
\end{split}
\end{align}
where $\eta$ is as in \eqref{exseq2} and
$d\mu$ is the tensor product of the standard measure on $\mathbb{R}^D$ with the measure on $H_1(\Gamma,\mathbb{R})$ induced by the morphism
$H_1(\Gamma,\mathbb{Z})\rightarrow H_1(\Gamma,\mathbb{R})$. This is the unique positive translation invariant measure on $H_1(\Gamma,\mathbb{R})$, with the 
property that 
a basis of $H_1(\Gamma,\mathbb{Z})$ generates a parallelogram of measure $1$. 
\end{deff}

Note that, with this definition, the amplitude is a function on $\mathbb{R}^{D(|V|-1)}$.

The Schwinger trick simply consists of using the identity $\int_0^{\infty} e^{-ax}dx = \frac{1}{a}$ in 
order to rewrite the amplitude in ``parametric form". For each edge $e$ we introduce a new variable $t_e$.

\begin{deff}
A subset of edges $S\subset E$  is called a spanning tree if the subgraph with edges in $S$ is a tree and is maximal in the sense that, if we add any of the remaining edges to it, it will contain a loop. 
Denote the set of spanning trees by $Span$.
\end{deff}

\begin{deff}
A subset of edges $C\subset E$  is called a cut if it has the following properties.
\begin{enumerate}
\item When we remove these edges, the graph becomes a disconnected union of trees. 
\item The set $C$ is minimal in the sense that, if we add back any edges to the remaining graph, it will 
either have a loop or become connected. 
\end{enumerate}
Since cuts are minimal, they divide vertices into disjoint sets 
$V_C$ and $V_C^c$. For a cut $C$ we denote by $P_C$ the norm 
of the sum of momenta in either component,
$$P_C= (\sum_{v\in V_C}p_v)^2=(\sum_{v\in V^c_C}p_v)^2.$$
Denote the set of cuts by $Cut$.
\end{deff}

\begin{deff}\label{First} The first Symanzik (or Kirchhoff) polynomial of a Feynman
graph is given by
$$ \Psi_{\Gamma}(t_1,...,t_{|E|}):=\sum_{S\in Span} \prod_{e\notin S} t_e .$$
\end{deff}

\begin{deff}\label{second} The second Symanzik polynomial of a Feynman graph is
given by
$$\mathcal{P}_{\Gamma}(t_1,...,t_{|E|},P_C):=\sum_{C \in Cut} P_C\prod_{e \in C} t_e.$$
\end{deff}

Given a Feynman graph $\Gamma$, we enumerate edges by $1,...,|E|$. 
We define
\begin{equation}\label{TTred}
T :={\rm diag}(\sqrt{t_1},\sqrt{t_2}, ..., \sqrt{t_{|E|}})\otimes Id_{D\times D}, \ \ \text{ and }
T_{{\rm red}} :={\rm diag}(\sqrt{t_1},...,\sqrt{t_{|E|}}). 
\end{equation}
Let $\vec{P}$ be a vector in $\mathbb{R}^{D\times|E|}$, 
where the coordinates are ordered in the same way as the variables $t_i$. 
Note that, for each edge, we have $D$ coefficients. Let $H$ denote the image 
of $H_1(\Gamma,\mathbb{R})$ in $\mathbb{R}^{|E|}$.

\begin{lem}\label{measurelemma}
The measure $d\nu$ on $T\beta^{-1}(a)$ induced by the standard measure 
on $\mathbb{R}^{D|E|}$ satisfies
\begin{equation}\label{dnu}
d\nu= \Psi_{\Gamma}(t_1,...,t_{|E|})^{D/2}(T\eta+T\tilde{a})_*(d\mu),
\end{equation}
with $\Psi_\Gamma$ the Kirchhoff polynomial, 
$T$ defined as in \eqref{TTred}, and $\eta$ as in \eqref{exseq2}.
\end{lem}

\begin{proof}
Since $T\eta+T\tilde{a}$ is linear, $(T\eta+T\tilde{a})_*(d\mu)$ is a constant multiple of the measure
on $T\beta^{-1}(a)$ induced from $\mathbb{R}^{D|E|}$.
In order to compute this constant, we can compare the volume of the image of the standard cube 
in these two measures. We have $(T\eta+T\tilde{a})_*(d\mu)\left(T\eta({\rm Cube})+T\tilde{a}\right)=d\mu({\rm Cube})=1$, by definition. On the other hand, we have 
$$d\nu\left(T\eta({\rm Cube})+T\tilde{a}\right)=d\nu\left(T\eta({\rm Cube})\right),$$ 
since the standard measure on $\mathbb{R}^{D|E|}$ is translation invariant. 

We choose bases $\{v_1,...v_{\ell}\}$ for $H_1(\Gamma,\mathbb{Z})$ and $\{A_1,...,A_D\}$ 
for $\mathbb{R}^D$. With $\eta$ as in \eqref{exseq2}, $\eta'$ as in \eqref{exseq1}, and
$T_{{\rm red}}$ as in \eqref{TTred}, the volume of the image of the standard cube is then given by
$$\sqrt{\det\big(T\eta (v_i\otimes a_j) \cdot T\eta (v_k\otimes a_l)\big)}=\big( \det(T_{red}\eta' v_i \cdot T_{red}\eta' v_k) \big)^{D/2},$$
since the volume form corresponding to a metric $g$ is $det(g)^{1/2}$. 
Consider $T_{{\rm red}} \eta': H_1(\Gamma,\mathbb{R})\rightarrow \mathbb{R}^{|E|}$.
We have $$\wedge^{\ell}T_{{\rm red}} \eta' :  \wedge^{\ell} H_1(\Gamma,\mathbb{R})\rightarrow \wedge^{\ell}\mathbb{R}^{|E|},$$ where $\ell$ is the dimension of $H_1(\Gamma,\mathbb{R})$. 
The determinant above is the norm square of $\wedge^{\ell}T_{{\rm red}} \eta'(v_1\wedge v_2 \wedge ... \wedge v_{\ell})$ in the induced metric, hence it can be computed as the sum of the squares of the
coefficients in an orthonormal basis. Let $\{w_1,...,w_{|E|}\}$ be the standard basis for $\mathbb{R}^{|E|}$. Then $\{\wedge_{i\in I} w_i\}_{|I|=\ell,I \subset \{1,..,|E|\}}$ is an orthonormal basis for 
$\wedge^{\ell}\mathbb{R}^{|E|}$. We have
\begin{align*}
\wedge^{\ell}T_{{\rm red}} \eta'(v_1\wedge v_2 \wedge ... \wedge v_{\ell}) &=(T_{{\rm red}} \eta'(v_1)\wedge T_{{\rm red}} \eta'(v_2) \wedge ... \wedge T_{{\rm red}} \eta'(v_{\ell})) \\&= (\sum_{j=1}^{|E|} \eta'_{1j}\sqrt{t_j}w_j  \wedge  \sum_{j=1}^{|E|} \eta'_{2j} \sqrt{t_j} w_j \wedge ... \wedge \sum_{j=1}^{|E|} \eta'_{\ell j}\sqrt{t_j}w_j)) .
\end{align*}

Note that, since $\eta'_{i,j}=\pm 1$, the coefficient of the term $w_{i_1}\wedge w_{i_2}\wedge... \wedge w_{i_{\ell}}$ is either zero or equal to $\pm \prod_k \sqrt{t_{i_k}}$. On the other hand, the coefficient of $w_{i_1}\wedge w_{i_2}\wedge... \wedge w_{i_{\ell}}$ is nonzero iff the orthogonal projection onto the subspace $W= span(w_{i_1}, w_{i_2}, ...,  w_{i_{\ell}})$ is an isomorphism when we restrict it to the image $T_{red}\eta'$. Since $T'$ fixes the coordinate subspaces, this map is an isomorphim iff 
$\text{Im}(\eta')\cap W= 0$, and that happens iff the subgraph with edges $i_1,...,i_\ell$ 
does not have a loop, which means it is the complement of a spanning tree.
To summarize, we have
$$\big( \det(T_{red}\eta' v_i \cdot T_{red}\eta' v_k) \big)^{D/2} =\big( \sum_{S \in Span} \prod_{e \notin S} t_e \big)^{D/2}. $$
\end{proof}

\begin{prop}\label{AmplProp}
When the integral \eqref{Ampl} converges, it is equal to
$$  \int_{\mathbb{R}^{|E|}_+} e^{-\sum_e t_e m_e^2}\prod_{e\in E} dt_e \frac{1}{\Psi_{\Gamma}(t_1,...,t_{|E|})^{D/2}}\int_{T\beta^{-1}(a)} e^{-\vec{P}.\vec{P} } d\nu, $$
where $d\nu$ is the measure on $T\beta^{-1}(a)$  induced by the standard measure 
on $\mathbb{R}^{D|E|}$.
\end{prop}

\medskip

\begin{proof}
Using the Schwinger trick $\int_0^{\infty} e^{-ax}dx = \frac{1}{a}$ we write
$$ \int_{\beta^{-1}(a)} \prod_e \frac{1}{Q_e} (\eta+\tilde{a})_*(d\mu) = \int_{\beta^{-1}(a)} \left( \int_{\mathbb{R}^{|E|}_+} e^{-\sum_e t_eQ_e}\prod_{e\in E} dt_e \right) (\eta+\tilde{a})_*(d\mu). $$
By the definition of $Q_e$, this is equal to
$$ \int_{\beta^{-1}(a)} \left((\int_{\mathbb{R}^{|E|}_+} e^{-\sum_e t_e m^2_e-\sum_e t_e P_e^2 }\prod_{e\in E} dt_e \right) (\eta+\tilde{a})_* (d\mu). $$
In vector form, with $T$ as in \eqref{TTred}, this can be written equivalently as
$$ \int_{\beta^{-1}(a)} \left(\int_{\mathbb{R}^{|E|}_+} e^{-\sum_e t_e m^2_e-T\vec{P} \cdot T\vec{P} }\prod_{e\in E} dt_e \right) (\eta+\tilde{a})_*(d\mu). $$
Since all functions are positive, convergence is the same as absolute convergence 
and we can switch integrals.
This gives 
$$ \int_{\mathbb{R}^{|E|}_+} e^{-\sum_e t_e m_e^2}\prod_{e\in E} dt_e\left(\int_{\beta^{-1}(a)} e^{-T\vec{P} \cdot T\vec{P} } (\eta+\tilde{a})_*(d\mu) \right). $$
Using the fact that $(\eta+\tilde{a})^* (e^{-T\vec{P} \cdot T\vec{P} })
=(T\eta+T\tilde{a})^* (e^{-\vec{P} \cdot \vec{P} })$, we rewrite the above as
$$ \int_{\mathbb{R}^{|E|}_+} e^{-\sum_e t_e m_e^2}\prod_{e\in E} dt_e\left(\int_{T\beta^{-1}(a)} e^{-\vec{P} \cdot \vec{P} } (T\eta+T\tilde{a})_*(d\mu)\right). $$
Then applying the result of Lemma \ref{measurelemma} we obtain 
$$ \int_{\mathbb{R}^{|E|}_+} e^{-\sum_e t_e m_e^2}\prod_{e\in E} dt_e\left(\int_{T\beta^{-1}(a)} e^{-\vec{P} \cdot \vec{P} } \frac{d\nu}{\Psi_{\Gamma}(t_1,...,t_{|E|})^{D/2}}\right). $$
\end{proof}

A standard computation shows the following simple facts.

\begin{lem}\label{Hlemma}
Let $H$ be a $d$-dimensional affine linear subspace in $\mathbb{R}^n$ and let $L$ 
be the distance of the affine subspace $H$ from the origin. 
Then the integral of a Gaussian function on $H$ 
with the induced measure is equal to $\pi^{d/2} e^{-L^2}$. 
\end{lem}

\begin{lem}\label{volumelemma}
Let $v_1,v_2,...,v_k$ be vectors in $\mathbb{R}^n$. One can compute the volume 
squared of the parallelogram generated by these vectors in the induced metric 
on the subspace they generate, in the form
$$|v_1 \wedge v_2\wedge...\wedge v_k|^2=\det(v_i \cdot v_j). $$
\end{lem}

The next statement then follows easily.

\begin{lem}\label{dislemma}
Suppose given a vector subspace $V$ of $\mathbb{R}^n$ with a basis $v_1,...,v_k$, and a vector
$a\in \mathbb{R}^n$. Let $P(w_1,...,w_m)$ denote the parallelogram generated by 
$w_1,...,w_m$. The distance of an affine subspace $a+V$ from the origin is equal to
$$ \frac{\text{Vol}(P(a,v_1,...,v_k))}{\text{Vol}(P(v_1,...,v_k))}
=\frac{|a\wedge v_1 \wedge ... \wedge v_k|}{|v_1 \wedge ... \wedge v_k|}.$$
\end{lem}

\begin{proof}
Volume is defined by the metric, hence, the distance times $Vol(P(v_1,...,v_k))$ is the volume of $P(a,v_1,...,v_k)$. 
\end{proof}


\begin{prop}\label{GaussintLemma} The Gaussian integral of $e^{-\vec{P}\cdot \vec{P}}$,  
with respect to the measure $d\nu$ on $T\beta^{-1}(a)$ defined as above, is given by
$$ \int_{T\beta^{-1}(a)} e^{-\vec{P}\cdot \vec{P}}d\nu = \pi^{D\ell/2}e^{-\mathcal{P}_{\Gamma}(t)/\Psi_{\Gamma}(t,a)}, $$
where $\Psi_\Gamma$ and $\mathcal{P}_\Gamma$ are the two Symanzik polynomials.
\end{prop}

\begin{proof}
By lemma \ref{Hlemma}, 
it is enough to show that the distance squared of $T\beta^{-1}(a)$ from the origin is 
$$\mathcal{P}_{\Gamma}(t)/\Psi_{\Gamma}(t,a).$$
Let $v_1,v_2,...,v_{\ell}$ be a basis of $H_1(\Gamma,\mathbb{Z})$. We denote the image 
of these vectors in $H_1(\Gamma,\mathbb{Z})\otimes \mathbb{R}$ by the same notation. 
Note that the affine subspace over which we are integrating the Gaussian is parallel to the 
space generated by $T \eta \big(v_1 \otimes e_1, v_1 \otimes e_2, ... , v_1 \otimes e_D, v_2 \otimes e_1,..., v_2\otimes e_D,...,v_{\ell} \otimes e_D\big)$, where $\{ e_1,...,e_D \}$ is the standard basis 
for $\mathbb{R}^D$. As we have shown before we then have
$$T \eta (v_1 \otimes e_1)\wedge ... \wedge T \eta (v_{\ell} \otimes e_D)=$$
$$\left(\sum_{S\in Span} (\pm\prod_{e\notin S} \sqrt{t_e}) \wedge_{e\notin S} w_e \otimes e_1\right) \wedge \left(\sum_{S \in Span} (\pm\prod_{e\notin S} \sqrt{t_e}) \wedge_{e\notin S} w_e \otimes e_2\right)\wedge $$ $$  ...   \wedge \left(\sum_{S\in Span} (\pm\prod_{e\notin S} \sqrt{t_e}) \wedge_{e\notin S} w_e \otimes e_D\right), $$
where $\{w_e\}_{e \in E}$ is a basis for $\mathbb{R}^{|E|}$.
We can write $\tilde{a} = \sum_{e,i} P_{e,i} w_e \otimes e_i$. Then, for  $i\neq j$, we see that the vector
$$T (w_e \otimes e_i) \wedge T \eta (v_1 \otimes e_1)\wedge ... \wedge T \eta (v_{\ell} \otimes e_D)$$ 
is orthogonal to 
$$T (w_{e'} \otimes e_j) \wedge T \eta (v_1 \otimes e_1)\wedge ... \wedge T \eta (v_{\ell} \otimes e_D).$$
One can see this from the expansion in the standard basis: all terms in the first expression 
have $\ell+1$ terms with $e_i$, while the second one has $\ell$ terms with $e_i$. Thus,
one can compute the norm squared of $$T (\sum _e P_{e,i} w_e \otimes e_i) \wedge T \eta (v_1 \otimes e_1)\wedge ... \wedge T \eta (v_{\ell} \otimes e_D)$$
for different $i$'s and add them up to get the squared norm of
$$T(\tilde{a})\wedge T \eta (v_1 \otimes e_1)\wedge ... \wedge T \eta (v_{\ell} \otimes e_D).$$

The vector above is in $\wedge^{D\ell +1} \mathbb{R}^{|E|}\otimes \mathbb{R}^D$. We can 
identify $\mathbb{R}^{|E|}\otimes \mathbb{R}^D$ with $D$ copies of $\mathbb{R}^{|E|}$. 
The norm squared of this vector is equal to the volume squared of the parallelogram 
generated by the vectors. Since we have $\ell+1$ vectors in one of the copies and $\ell$ vectors in the
other copies, we can compute the volume of each of them and multiply them together. For $j\neq i$ 
we have $v_1\otimes e_j,...,v_{\ell}\otimes e_j$, all of which have the same volume squared, 
equal to $\Psi$. To compute the volume of the copy with $\ell +1$ vectors, it is enough to compute
$$ T_{{\rm red}} (\sum_e P_{e,i} w_e) \wedge T_{{\rm red}} \eta' (v_1)\wedge ... \wedge T_{{\rm red}} \eta' (v_{\ell} )$$
$$= \left(\sum_e P_{e,i} \sqrt{t_e}\, w_e \right) \wedge  \left(\sum_{S\in Span} (\pm\prod_{e\in s} \sqrt{t_e}) \wedge_{e\in s} w_e \right).
$$
The terms that appear in the coefficients in the standard basis $$\{w_{e_{i_1}}\wedge...\wedge w_{e_{i_{\ell+1}}} \}_{i_1<...<i_{\ell+1}}$$ are sums of $P_{e,i}$. Nonzero terms correspond to $\ell+1$ edges that are a complement of a spanning tree plus one extra edge. Note that, if we remove these edges from the graph, it becomes disconnected and, if we add any of these edges to the graph, it becomes a spanning tree. So the term $w_{e_{i_1}}\wedge...\wedge w_{e_{i_{\ell+1}}}$ appears $\ell+1$ times, once of each of its edges. 
Thus, the coefficient is $\prod_j \sqrt{t_{e_{j}}} \sum_j P_{e_j,i}$. Since $P_{e_j,i}$ is a lift of $a$, this 
sum is equal to the sum of momenta in one of the connected components of the graph. We get the norm as
$$\mathcal{P}_{\Gamma}(t,P_C)=\sum_{C \in Cut} P_C\prod_{e \in C} t_e.$$
The original norm squared we wanted to compute is then
 $$\mathcal{P}_{\Gamma} \Psi_{\Gamma}^{D-1}$$
and, by Lemma \ref{dislemma}, 
the distance squared is 
$$\frac{\mathcal{P}_{\Gamma} \Psi_{\Gamma}^{D-1}}{\Psi_{\Gamma}^{D}}=\frac{\mathcal{P}_{\Gamma}}{\Psi_{\Gamma}}.$$
\end{proof}

The following result then follows from Proposition \ref{GaussintLemma} and Proposition \ref{AmplProp}.

\begin{prop}
When the integral \eqref{Ampl} converges, it is equal to
$$\mathcal{A}(\Gamma,a)=\pi^{D\ell/2}\int_{\mathbb{R}^{|E|}_+}e^{{-\sum t_e m_e^2}-\frac{\mathcal{P}_{\Gamma}}{\Psi_{\Gamma}}}\frac{1}{\Psi_{\Gamma}^{D/2}} \prod_e dt_e$$

\end{prop}

\begin{rem}\label{Pluckerremark}
The coefficients of the second Symanzik polynomials are always positive. According to the computation in Proposition \ref{GaussintLemma}, they correspond to squared norms of certain differential forms. We will use this property in Section \ref{feynGKZ}.
\end{rem}

\section{\textbf{Amplitude as a function of momenta}}\label{sec1}

The amplitude is defined as an integral which depends on the external momenta. This integral does not always converge. For some graphs, that are called ultraviolet divergent, the integral diverges for any value of external momenta, while for some graphs the divergences happen only for special values of the external momenta. The most widely used method in physics for treating these divergences is called dimensional regularization. Within this method, a regularization of divergent integrals is achieved by formally computing the integral for $D$ a complex variable in a neighborhood of the integer spacetime dimension in the complex plane. For a detailed explanation of this method see \cite{MC}. In this section, we define the integral for any $D$ and find differential equations satisfied by it.
\\

The integral depends on a parameter in $\mathbb{C}^{D(|V|-1)}$ and masses of edges. We map this vector space into the vector space $\mathfrak{V}_{\Gamma}$ which parametrizes the coefficients of  the first and second Symanzik polynomials so that it agrees with the amplitude on the image. We generalize the integral to an integral which has $\mathfrak{V}_{\Gamma}$ as its parameter space. Note that all coefficients in the second Symanzik polynomial are equal to $1$. We consider general coefficients for these terms and look at the integral as we vary them. Over $\mathfrak{V}_{\Gamma}$ the differential equation satisfied by the new integral is geometric in nature and can be solved using series. One can identify $\mathfrak{V}_{\Gamma}$ with the parameter space of a family of hypersurfaces in toric varieties. The value of the integral for integer $D$ is a period of a relative motive defined by the complement of this hypersurface. From now on we consider the following polynomials:
\\

$$\Psi_{\Gamma}(t,P_S):=\sum_{S \in Span} P_S\prod_{e\notin S} t_e
$$
$$\mathcal{P}_{\Gamma}(t,P_C):=\sum_{C \in Cut} P_C\prod_{e \in C} t_e
$$
\begin{align}\label{pp}
Q_{\Gamma}(t,P_C,m_e)&=\mathcal{P}_{\Gamma}(t,P_C)+ (\sum_e t_e m_e^2) \sum_{S \in Span} \prod_{e\notin S} t_e
= \sum_{T} Q_T t^T
\end{align}
Here $T$ ranges over all monomials appearing in $Q_{\Gamma}$ and by $t^T$ we mean the monomial corresponding to $T$. The amplitude is
$$\mathcal{A}(\Gamma,a,m_e)=\pi^{D\ell/2}\int_{\mathbb{R}^{|E|}_+}e^{-Q_{\Gamma}/\Psi_{\Gamma}}\frac{1}{\Psi_{\Gamma}^{D/2}}\prod_e dt_e,$$
where coefficients of $\Psi_{\Gamma}$ are $1$ and coefficients of $Q_{\Gamma}$ come from equation (\ref{pp}) and are sums of masses and momentum variables $P_C$.

\begin{deff}{(Parameter Space)}
Let $\mathfrak{V}_{\Gamma}$ denote the parameter space for $Q_T$ and $P_S$. It is a complex vector space of dimension equal to the number of monomials in $Q_{\Gamma}$ plus $|Span|$. 
\end{deff}

\begin{deff}{(Generalized Amplitude I)}\label{defGAI}
Given $(c_1,c_2,\vec{v})\in \mathbb{C}^{n+2}$, let
\begin{equation}
 I(c_1,c_2,Q_T,P_S,\vec{v}):=\int_{\mathbb{R}^{|E|}_+}e^{-{Q_{\Gamma}}/{\Psi_{\Gamma}}}\frac{Q_{\Gamma}^{c_1}}{\Psi_{\Gamma}^{{c_2}}}t^{\vec{v}}\prod_e dt_e,
\end{equation}
where $t^{\vec{v}}$ means $t_1^{v_1}t_2^{v_2}...t_n^{v_n}$ and $n=|E|$.
\end{deff}

\begin{rem}
Note that this integral is not well defined for all values of the parameters and is a (multi-valued) function on a dense domain in $\mathbb{C}^{n+2}\times \mathfrak{V}_{\Gamma}$. To define this (multi-valued) function first we define it for some open subset and then we take the analytic continuation.
\end{rem}

\begin{lem}\label{Radial}
When the generalized amplitude I converges we have:
$$ I(c_1,c_2,Q_T,P_S,\vec{v})= \Gamma(n+|\vec{v}|+c_1(\ell+1)-c_2\ell) \int_{\Delta_{n-1}} \frac{Q_{\Gamma}^{-n-|\vec{v}|+(c_2-c_1)\ell}}{\Psi_{\Gamma}^{-n-|\vec{v}|+(c_2-c_1)(\ell+1)}}t^{\vec{v}}  \Omega,
$$
where $|\vec{v}|=\sum_i v_i$, 
$$\Omega=\sum _{i=1}^n (-1)^{i+1} t_i \ dt_1 \wedge ... \wedge \Hat{dt_i} \wedge ... \wedge  dt_n, $$ 
and $\Delta_{n-1}$ is the standard $n-1$ simplex embedded in $\mathbb{R}^n$.
\end{lem}

\begin{proof}
To show this we parametrize $\mathbb{R}^{|E|}_+$ with $\Delta_{n-1}\times \mathbb{R}_+$. Consider the map $\phi: \Delta_{n-1}\times \mathbb{R}_+ \rightarrow \mathbb{R}^{|E|}_+$ given by $$\phi(x,s)= sx.$$
Note that $$\phi^*(dt_1dt_2...dt_n)= \Omega|_{\Delta_{n-1}} s^{n-1} ds,$$  $$Q_{\Gamma}(sx)= s^{\ell+1} Q_{\Gamma}(x),$$  $$\Psi_{\Gamma}(sx)= s^{\ell} \Psi_{\Gamma}(x)$$ and $(sx)^{\vec{v}}=s^{|\vec{v}|} x^{\vec{v}}$. Pulling back the integrand to $\Delta_{n-1}\times \mathbb{R}_+$ we have:
\begin{align} 
I(c_1,c_2,Q_T,P_S,\vec{v})&= \int_{\Delta_{n-1}\times \mathbb{R}_+} e^{-s{Q_{\Gamma}(x)}/{\Psi_{\Gamma}(x)}}\frac{s^{c_1 (\ell+1)}Q_{\Gamma}^{c_1}(x)}{s^{c_2 \ell} \Psi_{\Gamma}^{{c_2}}(x)} s^{|\vec{v}|} x^{\vec{v}} \Omega s^{n-1} ds \\ &=  \int_{\Delta_{n-1}}  \Omega \frac{Q_{\Gamma}^{c_1}(x)}{\Psi_{\Gamma}^{{c_2}}(x)}x^{\vec{v}} \int_{\mathbb{R}_+} s^{c_1 (\ell+1)- c_2 \ell + |\vec{v}| + n-1} e^{-s{Q_{\Gamma}(x)}/{\Psi_{\Gamma}(x)}} ds.
\end{align}
The fact that $\int_{\mathbb{R}_+} e^{-s\lambda} s^x ds = \lambda^{-x-1} \Gamma(x+1)$ then 
implies the lemma.
\end{proof}

\begin{deff}{(Generalized Amplitude II)}\label{defGAII}
Given $(c,d,\vec{v})\in \mathbb{C}^{n+2}$, let
\begin{equation}
 J(c,d,Q_T,P_S,\vec{v}):=\int_{\Delta_{n-1}} \frac{Q_{\Gamma}^c}{\Psi_{\Gamma}^d} t^{\vec{v}}  \Omega.
\end{equation}
\end{deff}

\begin{rem}
By Lemma \ref{Radial}, the generalized amplitudes I and II of Definitions \ref{defGAI} and \ref{defGAII} are related by \\
\begin{align}\label{ijrelation}
\begin{split}
I(c_1,c_2,Q_T,P_S,\vec{v})&=\Gamma(n+|\vec{v}|+c_1(\ell+1)-c_2\ell)\times \\ &J(-n-|\vec{v}|+(c_2-c_1)\ell,-n-|\vec{v}|+(c_2-c_1)(\ell+1),Q_T,P_S,\vec{v})
\end{split}
\end{align}
\end{rem}

Note that this is a function of $P_S$ and $Q_T$. For $S\in Span$, let $\vec{S}$ be the vector in $\mathbb{Z}^n$ with $1$ for edges that are not in $S$ and zero in the other coefficients. For a monomial $T=t_1^{\alpha_1}\cdots t_n^{\alpha_n}$, let $\vec{T}$ be the vector $(\alpha_1,\cdots,\alpha_n)$. We have
\begin{align}
\begin{split}
\frac{\partial}{\partial Q_T} I(c_1,c_2,Q_T,P_S,\vec{v})= c_1I(c_1-1,c_2,Q_T,P_S,\vec{v}+\vec{T}) \\
-I(c_1,c_2+1,Q_T,P_S,\vec{v}+\vec{T})  \label{diff1}
\end{split}
\end{align}
\begin{align}
\begin{split}
\frac{\partial}{\partial P_S} I(c_1,c_2,Q_T,P_S,\vec{v})=I(c_1+1,c_2+2,Q_T,P_S,\vec{v}+\vec{S})\\
 -c_2 I(c_1,c_2+1,Q_T,P_S,\vec{v}+\vec{S})\label{diff2}
\end{split}
\end{align}
Let $A\subset \mathbb{Z}^{n+1}$ be the set containing the following points. For each monomial $T$ in $Q_{\Gamma}$, consider the lattice point $(0,\vec{T})$ and, for any Spanning tree $S$, consider the lattice point $(1,\vec{S})$.\\

Denote the subset of $A$ of lattice points corresponding to spanning trees by $A_S$ and the subset of lattice points that correspond to monomials by $A_C=A\setminus A_S$. For $a=(1,\vec{S}) \in A_S$, $P_a$ refers to $P_S$ and for $a=(0,\vec{T}) \in A_C$,  $P_a$ refers to $Q_T$.\\

Since $Q_{\Gamma}$ is of degree $\ell+1$ and $\Psi_{\Gamma}$ is of degree $\ell$, all the lattice points lie on the affine hyperplane where the sum of the coordinates is $\ell+1$. Let $\phi:\mathbb{Z}^{n+1}\rightarrow\mathbb{Z}$ be the function that computes the sum of the coordinates and let $p_0:\mathbb{Z}^{n+1}\rightarrow\mathbb{Z}$ be the projection onto the first coordinate and $p_1:\mathbb{Z}^{n+1}\rightarrow\mathbb{Z}^n$ the projection onto the last $n$ coordinates. For any integer relation 
$\sum_{a\in A} n_a \vec{a}=0$ among lattice points in the set $A$, we have:
$$0=\phi(0)=\phi\left(\sum_{a\in A} n_a \vec{a}\right)=\sum_{a\in A} n_a \phi(\vec{a})=(\ell+1)\sum_{a\in A} n_a $$
$$0=p_0(0)=p_0 \left( \sum_{a\in A} n_a \vec{a} \right)=\sum_{a\in A} n_a p_0 (\vec{a})
=\sum_{a\in A_S} n_a. $$

Combining these two we have:
\begin{large}
\begin{align}\label{num}
\begin{cases}
\sum_{a\in A_S, n_a>0} n_a=\sum_{a\in A_S, n_a<0} -n_a\\
\sum_{a\in A_C, n_a>0} n_a=\sum_{a\in A_C, n_a<0} -n_a
\end{cases}
\end{align}
\end{large}

For the relation $(n_a)_{a\in A}$, we consider the following differential operator: 
$$\prod_{a:n_a>0} \big(\frac{\partial}{\partial P_a}\big)^{n_a}-\prod_{a:n_a<0}\big(\frac{\partial}{\partial P_a}\big)^{-n_a}.$$

\begin{prop}\label{GKZ1}
Let $A$ and $P_a$ be as above. For any $\mathbb{Z}$-linear relation $\sum_{a\in A} n_a \vec{a}=0$ we have 
\begin{align*}
&\left(\prod_{a:n_a>0} \big(\frac{\partial}{\partial P_a}\big)^{n_a}-\prod_{a:n_a<0}\big(\frac{\partial}{\partial P_a}\big)^{-n_a}\right) J(c,d,P_a,\vec{v})=0\ \\ &\left(\prod_{a:n_a>0} \big(\frac{\partial}{\partial P_a}\big)^{n_a}-\prod_{a:n_a<0}\big(\frac{\partial}{\partial P_a}\big)^{-n_a}\right)I(c_1,c_2,P_a,\vec{v})=0
\end{align*}
\end{prop}

\begin{proof}
Consider the set $Z=\{(c_1,c_2)+\mathbb{Z}^2\}\subset \mathbb{C}^2$. Each time we apply a derivation \eqref{diff1} or \eqref{diff2} to $I(c_1,c_2,Q_T,P_S,\vec{v})$, we get a weighted sum of two $I(c'_1,c'_2,Q_T,P_S,\vec{v}+p_1(a))$ where $(c'_1,c'_2)\in Z$. If we apply the positive part of the differential operator above, we get a weighted sum of $I(c'_1,c'_2,Q_T,P_S,\vec{v}+p_1(\sum_{a\in A, a>0} n_a a))$. If we apply the negative part of differential operator above, we get a weighted sum of $I(c'_1,c'_2,Q_T,P_S,\vec{v}+p_1(\sum_{a\in A, a<0} -n_a a))$. To show that the generalized amplitude goes to zero under this differential operator, it is enough to show that the weights are the same for positive and negative parts. Note that \eqref{num} implies that we apply each type (derivation with respect to $Q_T$ or $P_S$) of derivation the same number of times on both sides. Hence it is enough to show that $\delta_{(c_1,c_2)} \mapsto c_1\delta_{(c_1-1,c_2)}-\delta_{(c_1,c_2+1)}$ commutes with $\delta_{(c_1,c_2)} \mapsto \delta_{(c_1+1,c_2+2)}-c_2\delta_{(c_1,c_2+1)}$, which can be verified by direct inspection. A 
similar argument works for $J(c,d,P_a,\vec{v})$. \\
\end{proof}

\begin{prop}\label{GKZ2}
For each $i=0,...,n$ and for $a\in A$, let $a_i$ be the $i$-th coefficient of $a$. Assume that the generalized amplitude I converges. We have:
\begin{equation*}
\left(\sum_{a\in A} a_i P_a \frac{\partial}{\partial P_a}\right)I(c_1,c_2,P_a,\vec{v}) = \begin{cases}
i\neq0 &(-1-\vec{v}_i) I(c_1,c_2,P_a,\vec{v})\\
i= 0 &((c_1-c_2)(\ell+1)+\sum_i \vec{v}_i+n) I(c_1,c_2,P_a,\vec{v})
\end{cases}
\end{equation*}
\end{prop}

\begin{proof}
 For $i\neq0$ consider the action $(\alpha,P_a)\mapsto \alpha^{a_i}P_a$ of $\mathbb{G}_m$ on the 
 $P_a$'s. We want to see how the integral changes under this action. We have
\begin{align*}
&\int_{\mathbb{R}^{|E|}_+}e^{-\frac{Q_{\Gamma}(\alpha^{T_i} Q_T,t_1,...,t_n)}{\Psi_{\Gamma}(\alpha^{S_i} P_S,t_1,...,t_n)}}\frac{Q_{\Gamma}^{c_1}(\alpha^{T_i} Q_T,t_1,...,t_n)}{\Psi_{\Gamma}^{c_2}(\alpha^{S_i} P_S,t_1,...,t_n)}t^{\vec{v}} dt_1...dt_n = \\
&\int_{\mathbb{R}^{|E|}_+}e^{-\frac{Q_{\Gamma}(Q_T,t_1,...,\alpha t_i,...,t_n)}{\Psi_{\Gamma}( P_S,t_1,...,\alpha t_i,...,t_n)}}\frac{Q_{\Gamma}^{c_1}(Q_T,t_1,...,\alpha t_i,...,t_n)}{\Psi_{\Gamma}^{c_2}(P_S,t_1,...,\alpha t_i,...,t_n)}t^{\vec{v}} dt_1...dt_n = \\
&\int_{\mathbb{R}^{|E|}_+}e^{-\frac{Q_{\Gamma}}{\Psi_{\Gamma}}}\frac{Q_{\Gamma}^{c_1}}{\Psi_{\Gamma}^{c_2}}\frac{1}{\alpha^{\vec{v}_i+1}}t_1^{\vec{v}_1}...(\alpha t_i)^{\vec{v}_i}...t_n^{\vec{v}_n} dt_1...d(\alpha t_i)...dt_n .
\end{align*}
The last line is valid for multiplication by $\alpha$ real and positive which does not change the integration cycle hence it is also valid for all $\alpha$. 
As a result we have:
$$I(c_1,c_2,\alpha^{a_i} P_a,\vec{v})=\alpha^{-1-\vec{v}_i} I(c_1,c_2,P_a,\vec{v})
$$
Taking the derivative with respect to $\alpha$, evaluated at $\alpha=1$, we have:
\begin{align*}
\left(\sum_{a\in A} a_i P_a \frac{\partial}{\partial P_a}\right)	 I(c_1,c_2,P_a,\vec{v}) &=\frac{\partial}{\partial \alpha}I(c_1,c_2,\alpha^{a_i} P_a,\vec{v})\huge|_{\alpha=1}\\ &=(-1-\vec{v}_i) I(c_1,c_2,P_a,\vec{v})
\end{align*}

The other case we consider is when $i=0$. Note that $a_0$ is nonzero iff $a$ corresponds to a spanning tree. We scale all terms by $(\alpha,P_S)\mapsto \alpha^{a_0} P_S$ and we obtain
\begin{align*}
&\int_{\mathbb{R}^{|E|}_+}e^{-\frac{Q_{\Gamma}(Q_T,t_1,...,t_n)}{\Psi_{\Gamma}(\alpha P_S,t_1,...,t_n)}}\frac{Q_{\Gamma}^{c_1}(Q_T,t_1,...,t_n)}{\Psi_{\Gamma}^{c_2}(\alpha P_S,t_1,...,t_n)}t^{\vec{v}} dt_1...dt_n = \\
&\int_{\mathbb{R}^{|E|}_+}e^{-\frac{Q_{\Gamma}(Q_T,t_1,...,t_n)}{\alpha \Psi_{\Gamma}( P_S,t_1,...,t_n)}}\frac{Q_{\Gamma}^{c_1}(Q_T,t_1,...,t_n)}{\alpha^{c_2}\Psi_{\Gamma}^{c_2}( P_S,t_1,...,t_n)}t^{\vec{v}}. dt_1...dt_n =\\
\end{align*}
After setting $s_i=t_i/\alpha$ we obtain 
\begin{align*}
&\int_{\mathbb{R}^{|E|}_+}e^{-\frac{Q_{\Gamma}(Q_T,s_1,...,s_n)}{\Psi_{\Gamma}( P_S,s_1,...,s_n)}}\frac{\alpha^{c_1(\ell+1)-c_2\ell}Q_{\Gamma}^{c_1}(Q_T,s_1,...,s_n)}{\alpha^{c_2}\Psi_{\Gamma}^{c_2}( P_S,s_1,...,s_n)} \alpha^{\sum_i \vec{v}_i} s^{\vec{v}} \alpha^n ds_1...ds_n =
\\ 
&\alpha^{(c_1-c_2)(\ell+1)+\sum_i \vec{v}_i+n}\int_{\mathbb{R}^{|E|}_+}e^{-\frac{Q_{\Gamma}(Q_T,s_1,...,s_n)}{\Psi_{\Gamma}( P_S,s_1,...,s_n)}}\frac{Q_{\Gamma}^{c_1}(Q_T,s_1,...,s_n)}{\Psi_{\Gamma}^{c_2}( P_S,s_1,...,s_n)} s^{\vec{v}} ds_1...ds_n 
\end{align*}
As a result we have:
$$I(c_1,c_2,\alpha^{a_0} P_a,\vec{v})=\alpha^{(c_1-c_2)(\ell+1)+\sum_i \vec{v}_i+n} I(c_1,c_2,P_a,\vec{v})
$$
Taking the derivative with respect to $\alpha$, evaluated at $\alpha=1$, we then have:
\begin{align*}
\left(\sum_{a\in A} a_0 P_a \frac{\partial}{\partial P_a}\right) I(c_1,c_2,P_a,\vec{v}) &=\frac{\partial}{\partial \alpha}I(c_1,c_2,\alpha^{a_0} P_a,\vec{v})\huge|_{\alpha=1}\\ &=((c_1-c_2)(\ell+1)+\sum_i \vec{v}_i+n) I(c_1,c_2,P_a,\vec{v}).
\end{align*}
\end{proof}

\begin{prop}\label{GKZ3}
For each $i=0,...,n$ and for $a\in A$, let $a_i$ be the $i$-th coefficient of $a$. Assume that $c$ is positive, $d$ is negative and that all coefficients of $\vec{v}$ are positive. We have:
\begin{align*}
&\left(\sum_{a\in A} a_i P_a \frac{\partial}{\partial P_a}\right)J(c,d,P_a,\vec{v})\\ &= \begin{cases}
i\neq0 &  (-v_i-1)J(c,d,P_a,\vec{v})+(c(\ell+1)-d\ell+|v| +n)J(c,d,P_a,\vec{v}+e_i)\\
i= 0 &-d J(c,d,P_a,\vec{v}).
\end{cases}
\end{align*}
\end{prop}

\begin{proof}
The case $i=0$ can be dealt with easily by applying the scaling argument of the previous lemma. 
For other values of $i$, the scaling argument does not work, since scaling changes the integration cycle. 
Let $\theta=dt_1\wedge dt_2 ... \wedge dt_n$ and let $v$ be the vector field 
$\sum t_i \frac{\partial}{\partial t_i}$. Then it is not hard to see that  $\Omega = \iota_v \theta$.
Since $v-\frac{\partial}{\partial t_i}$ is tangent to $\Delta_{n-1}$, we have $\iota_v \theta|_{\Delta_{n-1}}=\iota_{\frac{\partial}{\partial t_i}} \theta|_{\Delta_{n-1}}$. We then obtain
\begin{align*}
\frac{\partial}{\partial \alpha}|_{\alpha =1} J(c,d,\alpha P_a,\vec{v}) &= \frac{\partial}{\partial \alpha}|_{\alpha =1} \int_{\Delta_{n-1}} \frac{Q_{\Gamma}^c(\alpha P_a)}{\Psi_{\Gamma}^d(\alpha P_a)} t^{\vec{v}}  \Omega \\ &=
\frac{\partial}{\partial \alpha}|_{\alpha =1} \int_{\Delta_{n-1}} \alpha^* \left(\frac{Q_{\Gamma}^c(P_a)}{\Psi_{\Gamma}^d(P_a)} t^{\vec{v}}\right) \alpha^{-v_i}  \Omega \\ &=\frac{\partial}{\partial \alpha}|_{\alpha =1} \int_{\Delta_{n-1}} \alpha^* \left(\frac{Q_{\Gamma}^c(P_a)}{\Psi_{\Gamma}^d(P_a)} t^{\vec{v}}\Omega\right) \alpha^{-v_i-1} \\ &=  (-v_i-1)J(c,d,P_a,\vec{v}) + \int_{\Delta_{n-1}} \frac{\partial}{\partial \alpha}|_{\alpha =1} 
\,\, \alpha^* \left(\frac{Q_{\Gamma}^c(P_a)}{\Psi_{\Gamma}^d(P_a)} t^{\vec{v}}\Omega\right) \\
\frac{\partial}{\partial \alpha}|_{\alpha =1}\,\,  \alpha^* \left(\frac{Q_{\Gamma}^c(P_a)}{\Psi_{\Gamma}^d(P_a)} t^{\vec{v}}\Omega\right)&= \mathcal{L}_{t_i \frac{\partial}{\partial t_i}} \left(\frac{Q_{\Gamma}^c(P_a)}{\Psi_{\Gamma}^d(P_a)} t^{\vec{v}}\Omega\right)\\ &=(\iota_{t_i \frac{\partial}{\partial t_i}} \circ d+ d\circ \iota_{t_i \frac{\partial}{\partial t_i}}) \left(\frac{Q_{\Gamma}^c(P_a)}{\Psi_{\Gamma}^d(P_a)} t^{\vec{v}}\Omega\right) \\ &=
(\iota_{t_i \frac{\partial}{\partial t_i}} \circ d)\circ \iota_v \left(\frac{Q_{\Gamma}^c(P_a)}{\Psi_{\Gamma}^d(P_a)} t^{\vec{v}}\theta\right) + d  \left(\iota_{t_i \frac{\partial}{\partial t_i}} \frac{Q_{\Gamma}^c(P_a)}{\Psi_{\Gamma}^d(P_a)} t^{\vec{v}}\Omega\right)  \\ &= \iota_{t_i \frac{\partial}{\partial t_i}} \mathcal{L}_{v} \left(\frac{Q_{\Gamma}^c(P_a)}{\Psi_{\Gamma}^d(P_a)} t^{\vec{v}}\theta\right) + d  \left(\iota_{t_i \frac{\partial}{\partial t_i}} \frac{Q_{\Gamma}^c(P_a)}{\Psi_{\Gamma}^d(P_a)} t^{\vec{v}}\Omega\right) \\ &=deg\left(\frac{Q_{\Gamma}^c(P_a)}{\Psi_{\Gamma}^d(P_a)} t^{\vec{v}}\theta\right) \iota_{t_i \frac{\partial}{\partial t_i}}  \left(\frac{Q_{\Gamma}^c(P_a)}{\Psi_{\Gamma}^d(P_a)} t^{\vec{v}}\theta\right) + d  \left(\iota_{t_i \frac{\partial}{\partial t_i}} \frac{Q_{\Gamma}^c(P_a)}{\Psi_{\Gamma}^d(P_a)} t^{\vec{v}}\Omega\right)
\end{align*}
Let $\Sigma$ be the union of the coordinate hyperplanes. We have $$\frac{Q_{\Gamma}^c(P_a)}{\Psi_{\Gamma}^d(P_a)} t^{\vec{v}}\Omega  |_{\Sigma}=0$$ and hence  $$\iota_{t_i \frac{\partial}{\partial t_i}}(\frac{Q_{\Gamma}^c(P_a)}{\Psi_{\Gamma}^d(P_a)} t^{\vec{v}}\Omega) |_{\Sigma}=0.$$ Here we are using the fact that the vector field $t_i \frac{\partial}{\partial t_i}$ is tangent to $\Sigma$, so that we can first restrict to $\Sigma$ and then perform the contraction. Since the boundary of $\Delta_{n-1}$ lies on $\Sigma$, the integral of the second term vanishes.\\

Using $\iota_v \theta|_{\Delta_{n-1}}=\iota_{\frac{\partial}{\partial t_i}} \theta|_{\Delta_{n-1}}$, we can compute the integral of the first term as
\begin{align*}
\int_{\Delta_{n-1}} \frac{\partial}{\partial \alpha}|_{\alpha =1}  \,\, \alpha^* \left(\frac{Q_{\Gamma}^c(P_a)}{\Psi_{\Gamma}^d(P_a)} t^{\vec{v}}\Omega\right) &= \int_{\Delta_{n-1}}  deg\left(\frac{Q_{\Gamma}^c(P_a)}{\Psi_{\Gamma}^d(P_a)} t^{\vec{v}}\theta\right) t_i \iota_{\frac{\partial}{\partial t_i}} \left(\frac{Q_{\Gamma}^c(P_a)}{\Psi_{\Gamma}^d(P_a)} t^{\vec{v}}\theta\right)\\
&=  deg\left(\frac{Q_{\Gamma}^c(P_a)}{\Psi_{\Gamma}^d(P_a)} t^{\vec{v}}\theta\right) \int_{\Delta_{n-1}} \left(\frac{Q_{\Gamma}^c(P_a)}{\Psi_{\Gamma}^d(P_a)} t^{\vec{v}+e_i}\Omega\right)\\ &= (c(\ell+1)-d\ell+|v| +n)J(c,d,P_a,\vec{v}+e_i).
\end{align*}
Summing up we obtain
$$
\frac{\partial}{\partial \alpha}|_{\alpha =1} J(c,d,\alpha P_a,\vec{v})= (-v_i-1)J(c,d,P_a,\vec{v})+(c(\ell+1)-d\ell+|v| +n)J(c,d,P_a,\vec{v}+e_i).
$$
\end{proof}

\begin{lem}\label{analytic} The generalized amplitude 
$J(c,d,P_a,\vec{v})$, which is holomorphic for $\Re(c)>0$, $\Re(d)<0$ and $\Re(v)>0$, has an analytic continuation which is meromorphic on $\mathbb{C}^{n+2}$.  
\end{lem}

\begin{proof}
This basically follows from resolution of singularity and the following fact. Let $P_i(x)$ be polynomials in $n$ variables which are bounded away from zero on the hypercube $[0,1]^n$. One needs to show that the integral
$$
\int_{[0,1]^n} P_1(t)^{c} P_2(t)^{d} t_1^{a_1}t_2^{a_2}\cdots t_n^{a_n}  dt_1\cdots dt_n,
$$
which is defined and holomorphic in  $\{ \Re(a_i)>0\}$, has an analytic continuation to $\mathbb{C}^{n+2}$. We prove this by induction on $n$. The base case is the observation that $C^c D^d$ has analytic continuation for $C$ and $D$ nonzero, which is clearly true. Note that we have
\begin{align*}
&\int_{[0,1]^n\cap t_1=\{0,1\}} (-1)^{t_1+1}\left(P_1(t)^{c} P_2(t)^{d} t_1^{a_1+1}t_2^{a_2}\cdots t_n^{a_n}\right)  dt_2\cdots dt_n
\\&=\int_{[0,1]^n} \frac{\partial}{\partial t_1} \left(P_1(t)^{c} P_2(t)^{d} t_1^{a_1+1}t_2^{a_2}\cdots t_n^{a_n}\right)  dt_1\cdots dt_n\\&=
\int_{[0,1]^n} \frac{\partial}{\partial t_1} \left(P_1(t)^{c} P_2(t)^{d}\right) t_1^{a_1+1}t_2^{a_2}\cdots t_n^{a_n}  dt_1\cdots dt_n 
\\&+ \int_{[0,1]^n} P_1(t)^{c} P_2(t)^{d} (a_1+1)t_1^{a_1}t_2^{a_2}\cdots t_n^{a_n}  dt_1\cdots dt_n \\
&= \int_{[0,1]^n}  \left( c P_1(t)^{c-1} \frac{\partial P_1}{\partial t_1} P_2(t)^{d}+ d P_1(t)^{c} \frac{\partial P_2}{\partial t_1} P_2(t)^{d-1}\right) t_1^{a_1+1}t_2^{a_2}\cdots t_n^{a_n}  dt_1\cdots dt_n 
\\&+ \int_{[0,1]^n} P_1(t)^{c} P_2(t)^{d} (a_1+1)t_1^{a_1}t_2^{a_2}\cdots t_n^{a_n}  dt_1\cdots dt_n 
\end{align*}
Assume we have the analytic continuation for the region $\Re(a_i)>m_i$. Note that we also have it for $m_i=0$, since the $P_i$'s are nonzero on the hypercube. From the computation above we have:
\begin{align*}
&\int_{[0,1]^n} P_1(t)^{c} P_2(t)^{d} t_1^{a_1}t_2^{a_2}\cdots t_n^{a_n}  dt_1\cdots dt_n \\&= \frac{1}{a_1+1} \int_{[0,1]^n\cap t_1=\{0,1\}} (-1)^{t_1+1}\left(P_1(t)^{c} P_2(t)^{d} t_1^{a_1+1}t_2^{a_2}\cdots t_n^{a_n}\right)  dt_2\cdots dt_n \\ &+ \frac{1}{a_1+1}  \int_{[0,1]^n}  c P_1(t)^{c-1} P_2(t)^{d}\frac{\partial P_1}{\partial t_1} t_1^{a_1+1}t_2^{a_2}\cdots t_n^{a_n}  dt_1\cdots dt_n \\ &+
\frac{1}{a_1+1}  \int_{[0,1]^n} d P_1(t)^{c}  P_2(t)^{d-1}\frac{\partial P_2}{\partial t_1} t_1^{a_1+1}t_2^{a_2}\cdots t_n^{a_n}  dt_1\cdots dt_n .
\end{align*}

By induction, the first term is meromorphic. The second and third terms have analytic continuation to the region $\Re(a_1)+1>m_1$ and $\Re(a_i) \geq m_i$. Thus, we have analytic continuation to the region 
$\Re(a_1)> m_1-1$.  We can continue this for all $m_i$ and prove it for any value of $a_i$.
\\

Now, using a special case of resolution of singularities, we can rewrite any integral over $\Delta_{n-1}$ as a sum of integrals of the form above. This is the same argument that is used in \cite{BW} 
Theorem 2, so we do not repeat it explicitly here. 
\end{proof}

By Remark \ref{ijrelation}, we can define $I(P_a)$ using the corresponding value of $J(P_a)$, i.e.
$$
I(c_1,c_2,Q_T,P_S,\vec{v})=CJ(-n-|\vec{v}|+(c_2-c_1)\ell,-n-|\vec{v}|+(c_2-c_1)(\ell+1),Q_T,P_S,\vec{v})
$$
where C is constant. Hence we see that the second term in the case $i\neq 0$ of proposition \ref{GKZ3} vanishes. 
\begin{align*}
&c= -n-|\vec{v}|+(c_2-c_1)\ell \\
&d= -n-|\vec{v}|+(c_2-c_1)(\ell+1)\\
&c(\ell+1)-d\ell+ |\vec{v}| + n = 0
\end{align*}

\begin{theorem}\label{main}
Let $w_0=(c_1,c_2,\vec{v})$ and $w=(x,y,\vec{u})$ be any vectors in $\mathbb{C}^{n+2}$. One can pull back $I(P_a,w_0+\epsilon w)$ to a neighborhood of $\epsilon=0$ and take the Laurent expansion. Assume that the Laurent expansion has the following form:
$$I(P_a,w_0+\epsilon w)= \sum_{i\geq -n} \epsilon^i I_i(P_a) .$$ 
Then we have:
$$\left(\prod_{a:n_a>0} \big(\frac{\partial}{\partial P_a}\big)^{n_a}-\prod_{a:n_a<0}\big(\frac{\partial}{\partial P_a}\big)^{-n_a}\right)I_i(P_a)=0$$
and 
\begin{align*}
\left(\sum_{a\in A} a_0 P_a \frac{\partial}{\partial P_a}\right) I_i(P_a)&=((\ell+1,-\ell-1,1,\cdots,1) \cdot w_0+n) I_i(P_a) \\&+ (\ell+1,-\ell-1,1,\cdots,1) \cdot w   \ I_{i-1}(P_a)
\end{align*}
and for $k\neq 0$
\begin{align*}
\left(\sum_{a\in A} a_k P_a \frac{\partial}{\partial P_a}\right) I_i(P_a)&=(-1-e_k \cdot \vec{v}) I_i(P_a) - e_k \cdot \vec{u}   \ I_{i-1}(P_a)
\end{align*}
\end{theorem}

\begin{proof}
Our definition of the integral is by analytic continuation of $J(c,d,P_a,\vec{v})$. Note that for any differential operator $L$, $L(J)$ is meromorphic. Since the differential equation is satisfied for an open subset of $\mathbb{C}^{n+2}$ ($\Re(c)>0$,$\Re(d)<0$ and $\Re(\vec{v})>0$), it is valid for all values of $c,d$ and $\vec{v}$. Note that $I(P_a,w_0+\epsilon w)$ for all values of $\epsilon$ satisfies the first equation since $J$ has the same property. By the calculation above we see that the second term in the case $k\neq 0$ 
vanishes and we have 
\begin{align*}
\left(\sum_{a\in A} a_k P_a \frac{\partial}{\partial P_a}\right) I(P_a,w_0+\epsilon w)&=(-1-e_k.(\vec{v}+\epsilon \vec{u}))\  I(P_a,w_0+\epsilon w) \\ &=(-1-e_k.\vec{v}) I(P_a,w_0+\epsilon w) - \epsilon\  e_k.\vec{u} \ I(P_a,w_0+\epsilon w) .
\end{align*}
In the case $k=0$ we have
\begin{align*}
\left(\sum_{a\in A} a_0 P_a \frac{\partial}{\partial P_a}\right) I(P_a,w_0+\epsilon w)&=(n+|\vec{v}+\epsilon \vec{u}|-(c_2+\epsilon y -c_1 -\epsilon x)(\ell+1))\  I(P_a,w_0+\epsilon w) \\ &=((\ell+1,-\ell-1,1,\cdots,1).w_0+n)  I(P_a,w_0+\epsilon w) \\ &+
\epsilon\  (\ell+1,-\ell-1,1,\cdots,1).w \   I(P_a,w_0+\epsilon w) .
\end{align*}
The theorem follows from expanding $I(P_a,w_0+\epsilon w)$ and comparing terms with different 
powers of $\epsilon$.

\end{proof}
\begin{rem}
In standard dimensional regularization for the amplitude, assuming that we take the expansion with respect to $\epsilon$ in $D/2+\epsilon$, i.e.
$$
\sum_{i\geq -n}\epsilon^i I_i(P_a)= \int_{\mathbb{R}_{+}^{|E|}} e^{-Q_{\Gamma}/\Psi_{\Gamma}} \frac{1}{\Psi_{\Gamma}^{D/2+\epsilon}},
$$
we have  $w_0=(0,D/2,\vec{0})$ and $w=(0,1,\vec{0})$, hence the amplitude satisfies the differential equations
$$\left(\prod_{a:n_a>0} \big(\frac{\partial}{\partial P_a}\big)^{n_a}-\prod_{a:n_a<0}\big(\frac{\partial}{\partial P_a}\big)^{-n_a}\right)I_i(P_a)=0$$
\\
$$
\left(\sum_{a\in A} a_0 P_a \frac{\partial}{\partial P_a}\right) I_i(P_a)=(n-(\ell+1)D/2) I_i(P_a) - (\ell+1) I_{i-1}(P_a)
$$
and for $k\neq 0$ 
$$
\left(\sum_{a\in A} a_k P_a \frac{\partial}{\partial P_a}\right) I_i(P_a)=- I_i(P_a) \, .$$
\\
\end{rem}
In particular the lowest coefficient satisfies the so called GKZ hypergeometric differential 
equation, which we consider in the next section.

\section{\textbf{GKZ A-Hypergeometric Differemtial Equations}}\label{GKZsection}
\begin{deff}\label{phi}
Given $\mathbb{Z}^n$, the $n$-dimensional lattice, we fix a basis and denote an element as $n$-tuple of integers. We define $\phi_i$ as the map from $\mathbb{Z}^n$ to $\mathbb{Z}$ which gives us the $i$-th coordinate in the fixed basis.
\end{deff}

Let $A=\{a_1,\cdots,a_N\} \subset \mathbb{Z}^n$ be a set of lattice points such that they all lie in the hyperplane $\phi_1=1$ and generate the lattice as a $\mathbb{Z}$ module. For a tuple of integers $r=(n_a: a \in A)$ consider the relation among the points of $A$ of the form 
$$\sum_{a\in A} n_a a=0$$
Denote the set of relations by $R$. For each $r\in R$, we consider a corresponding differential operator
\begin{equation}\label{box}
\Box_r :=\prod_{\substack{a\in A\\ n_a>0}} \left(\frac{\partial}{\partial p_a}\right)^{n_a} -\prod_{\substack{a\in A\\ n_a<0}} \left(\frac{\partial}{\partial p_a}\right)^{-n_a}
\end{equation}
and for $i=1,\cdots,n$ we define
\begin{equation}
\label{zi}
Z_i:= \sum_{a\in A} \phi_i(a) p_a \frac{\partial}{\partial p_a}    \, .
\end{equation}

 On $V=\mathbb{C}^N$, with coordinates $p_1,...,p_N$,
 consider the differential equations 
$$
\Box_r \phi=0  \, .
$$

For $\beta=(\beta_1,...,\beta_n) \in k^n \subset \mathbb{C}^n$, consider the differential equations
$$
(Z_i- \beta_i) \phi=0  \,  .
$$

We want to find solutions to these differential equations. We denote by $W$ the Weyl algebra 
$$W=k[p_a,\frac{\partial}{\partial p_a}: a\in A]/ ([\frac{\partial}{\partial p_a},p_b]=\delta^a_b,[p_a,p_b]=0,[\frac{\partial}{\partial p_b},\frac{\partial}{\partial p_a}]=0) , $$
where $k$ is a sub-field of $\mathbb{C}$. Then the GKZ left $W$-module is defined by 
\\
$$ H_A(\beta) =W / \sum_i W(Z_i-\beta_i) + \sum_r W\ \Box_r  \, .
$$
\begin{rem}
One can consider $H_A(\beta)$ as a $D_V$-module, i.e. as a sheaf of modules over the sheaf of differential operators on $V$. $W$ is the ring of global sections of $D_V$ and $H_A(\beta)$ is the space of global section of the corresponding sheaf.
\end{rem}
This set of differential equations and the corresponding $W$-module was considered by Gelfand, Kapranov and Zelevinsky in (\cite{GKZ}, \cite{GKZH}). For a complete discussion of results in this direction see \cite{Nook}, \cite{Saito} and the references there. We prove the relevant results and state the theorems that we need. Their main results can be summarized in the following theorem from \cite{cattani}. 
\
\begin{theorem}{(GKZ)}
Let $H_A(\beta)$ be a GKZ hypergeometric system.\\

$(1)$ $H_A(\beta)$ is always holonomic.\\

$(2)$ The singular locus of $H_A(\beta)$ is independent of $\beta\in \mathbb{C}^n$ and agrees with the zero locus of the principal $A$-determinant $E_A(x)$ defined in chapter $10$ of \cite{Nook}. 
\\

$(3)$ For arbitrary $A$ and generic $\beta$, the holonomic rank of $H_A(\beta)$ equals the normalized
volume of the convex hull of $A$, $vol(conv(A))$. \\

$(4)$ For arbitrary $A$ and $\beta$, $rank(H_A(\beta)) \geq vol(conv(A))$.\\

$(5)$ Given $A$, $rank(H_A(\beta)) = vol(conv(A))$ for all $\beta \in \mathbb{C}^n$
if and only if the toric ideal $I_A$ is
Cohen-Macaulay.

\end{theorem}

\begin{prop}
Let $t^a=t_1^{\phi_1(a)} t_2^{\phi_2(a)} \cdots t_n^{\phi_n(a)}$. Consider the ring 
$$\mathfrak{R}= k[p_a,t^a: a \in A]$$ with an action of $W$ given by 
$$(p_a,P)\mapsto p_a P$$
$$(\frac{\partial}{\partial p_a},P)\mapsto \frac{\partial}{\partial p_a} P + t^a P. $$
Moreover, consider the map
$$\Psi: W \rightarrow \mathfrak{R}$$
$$\Psi(p^I\prod_{a\in A}\left(\frac{\partial}{\partial p_a}\right)^{m_a})= p^I t^{\sum m_a a}\, . $$
Let 
$$f=\sum_{a\in A} P_a t^a$$ 
$$Y_i=  t_i \frac{\partial f}{\partial t_i} +  t_i \frac{\partial}{\partial t_i}  - \beta_i$$
Then we have:\\

(1) $\Psi$ is surjective\\

(2) $ker(\Psi)= W \sum_a \Box_a$, with $\Box_a$ as in (\ref{box})
\\

(3) $image(W (Z_i-\beta_i)) =  Y_i \mathfrak{R}$ , with $Z_i$ as in (\ref{zi})
\\

(4) $\Psi$ gives an isomorphism between $H_A(\beta)$ and $\mathfrak{R}/\sum_i Y_i\mathfrak{R} $.
\end{prop}
\begin{proof}
The first two statements follow from the definition. Note that $\Psi$ is $k[p_a:a\in A]$ linear and $Y_i$ acts $k[p_a:a\in A]$ linearly. Thus, to check (3) it is enough to compute the image of
$$ \prod_{a\in A}\left(\frac{\partial}{\partial p_a}\right)^{m_a} (Z_i-\beta_i)$$
With $\phi_i$ as Definition \ref{phi}, we have 
\begin{align*}
\Psi(\prod_{a\in A}\left(\frac{\partial}{\partial p_a}\right)^{m_a} (Z_i-\beta_i))&= \Psi(\prod_{a\in A}\left(\frac{\partial}{\partial p_a}\right)^{m_a} \sum_{b\in A} \phi_i(b) p_b \frac{\partial}{\partial p_b}-\beta_i\prod_{a\in A}\left(\frac{\partial}{\partial p_a}\right)^{m_a}) \\ &=  \sum_{b\in A} \phi_i(b) p_b \Psi(\frac{\partial}{\partial p_b} \prod_{a\in A}\left(\frac{\partial}{\partial p_a}\right)^{m_a} ) \\ &+ \sum_{b\in A} \phi_i(b) m_b  \Psi(\prod_{a\in A}\left(\frac{\partial}{\partial p_a}\right)^{m_a})-\beta_i \Psi (\prod_{a\in A}\left(\frac{\partial}{\partial p_a}\right)^{m_a}) \\ &=
\sum_{b\in A} \phi_i(b) p_b t^b t^{\sum m_a a} + \sum_{b\in A} \phi_i(b) m_b  t^{\sum m_a a}-\beta_i t^{\sum m_a a} \\ &= 
( \sum_{a\in A} \phi_i(a) p_a t^a + \phi_i(\sum_{a\in A} a m_a) -\beta_i) t^{\sum m_a a}\\
&= Y_i t^{\sum m_a a}
\end{align*}
To check (4) we need to show that, for $P \in GKZ$, we have 
$$\Psi(\frac{\partial}{\partial p_b} P)=  \frac{\partial}{\partial p_b} \Psi(P) + t^b \Psi(P)\, . $$
If $P$ has the form $p^I\prod_{a\in A}\left(\frac{\partial}{\partial p_a}\right)^{m_a}$, we have
\begin{align*}
\Psi(\frac{\partial}{\partial p_b} P)&= \Psi(\left(\frac{\partial}{\partial p_b}p^I\right)\prod_{a\in A}\left(\frac{\partial}{\partial p_a}\right)^{m_a}+p^I\prod_{a\in A}\left(\frac{\partial}{\partial p_a}\right)^{m_a}\frac{\partial}{\partial p_b} )\\
&= \frac{\partial}{\partial p_b} \Psi(P) + t^b \Psi(P)\, .
\end{align*}
\end{proof}

This construction makes GKZ a quotient of $\mathfrak{R}$. Note that it is not a $\mathfrak{R}$-module since the action of $Y_i$ does not commute with multiplication by $t^a$. However, it is a $k[p_a: a \in A]$ module. We want to understand the structure as a $k[p_a: a \in A]$-module. 

\begin{prop}
Assume $-\beta_1 + n$ is nonzero in $k$ for all $n \in \mathbb{Z}_{\geq 0}$. We have an isomorphism of $k[p_a: a \in A]$ modules 
\begin{equation}\label{betazero}
\mathfrak{R}/Y_1 \mathfrak{R} \cong k[p_a: a \in A][t^a/f] \cong k[p_a: a \in A][t^a]/ \left(\sum_{a\in A} p_a t^a =1\right)=:\hat{\mathfrak{R}} \, .
\end{equation}
Let $\phi_1$ be the first coordinate as Definition \ref{phi}. The isomorphism is given by 
$$\mathfrak{F}:\mathfrak{R} \rightarrow \mathfrak{R}$$
$$\mathfrak{F}(t^J)= (-1)^{\phi_1(J)}\gamma(\phi_1(j))\ t^J,$$
where $\gamma(n)= -\beta_1 (-\beta_1+1) \cdots (-\beta_1+n-1)$ and $\gamma(0)=\gamma(1)=-\beta_1$.
\end{prop} 
\begin{proof}
Note that we have $\gamma(n+1)/\gamma(n)= -\beta_1 +n$ and since all points of $A$ have $\phi_1=1$, we have $t_1 \frac{\partial f}{\partial t_1} =f$
\begin{align*}
\mathfrak{F}(Y_1 t^J)&=\mathfrak{F}(t_1 \frac{\partial f}{\partial t_1}t^J +  t_1 \frac{\partial}{\partial t_1} t^J - \beta_1 t^J)\\
&=\mathfrak{F}(\sum_{a\in A} p_a t^{J+a} + \phi_1(J) t^J - \beta_1 t^J)\\
&= \sum_{a \in A} p_a t^a (-1)^{\phi_1(J)+1}\gamma(\phi_1(J)+1)\ t^J \\&+ \phi_1(J)(-1)^{\phi_1(J)}\gamma(\phi_1(J))\ t^J- \beta_1 (-1)^{\phi_1(J)}\gamma(\phi_1(J))\ t^J\\
&=\left(-f \frac{\gamma(\phi_1(J)+1)}{\gamma(\phi_1(J))} + \phi_1(J) - \beta_1 \right)(-1)^{\phi_1(J)}\gamma(\phi_1(J)) t^J\\
&=(-\beta_1+\phi_1(J))(1-f)(-1)^{\phi_1(J)}\gamma(\phi_1(J)) t^J\\
&=(1-f)(-\beta_1+\phi_1(J)) \mathfrak{F}(t^J)
\end{align*}
By definition $\mathfrak{F}$ is surjective. The equations above shows that the image of $Y_1 \mathfrak{R}$ is the ideal generated by $1-f$ and we have the isomorphism.
\end{proof}
\begin{prop}\label{dwork}
Assume $-\beta_1+n$ is nonzero in $k$, for all $n \in \mathbb{Z}_{\geq 0}$. Define $\tilde{Y}_i$ by 
$$
\tilde{Y}_i: \proj \rightarrow \proj
$$
$$
\tilde{Y}_i = \left( t_i \frac{\partial}{\partial t_i} -\beta_i -t_i\frac{\partial f}{\partial t_i} (-\beta_1+t_1\frac{\partial}{\partial t_1})\right)  \, .
$$
We have
$$
H_A(\beta)= \proj/ \sum_{i=2}^n \tilde{Y}_i \proj  \, .
$$
\end{prop}
\begin{proof}
We need to find the image of $Y_i \mathfrak{R}$ under $\ef$. We have 
\begin{align*}
\ef(Y_i t^J)&= \ef(t_i \frac{\partial f}{\partial t_i}t^J +  t_i \frac{\partial}{\partial t_i} t^J - \beta_i t^J)\\
&=(-1)^{\phi_1(J)+1}\gamma(\phi_1(J)+1) t_i \frac{\partial f}{\partial t_i} t^J + (-1)^{\phi_1(J)}\gamma(\phi_1(J))t_i \frac{\partial}{\partial t_i} t^J-\beta_i (-1)^{\phi_1(J)}\gamma(\phi_1(J)) t^J\\
&=\left(-(-\beta_1+\phi_1(J))t_i\frac{\partial f}{\partial t_i}+ t_i \frac{\partial}{\partial t_i} -\beta_i\right) (-1)^{\phi_1(J)}\gamma(\phi_1(J)) t^J  \\
&=\left(-t_i\frac{\partial f}{\partial t_i} (-\beta_1+t_1\frac{\partial}{\partial t_1})+ t_i \frac{\partial}{\partial t_i} -\beta_i\right) (-1)^{\phi_1(J)}\gamma(\phi_1(J)) t^J   \, .
\end{align*}
Thus, we have that the image is 
$$ \left( t_i \frac{\partial}{\partial t_i} -\beta_i -t_i\frac{\partial f}{\partial t_i} (-\beta_1+t_1\frac{\partial}{\partial t_1})\right) \mathfrak{R} $$
\end{proof}
To show that the equations come from geometry we observe that the ring $\mathfrak{R}$ is the coordinate ring of an affine toric variety. Let $\mathbb{N}A$ be the semigroup generated by the set $A$ as a sub-semigroup of $\mathbb{Z}^n$. By definition, $\mathfrak{R}$ is the semigroup algebra $k[p_a:a\in A][\mathbb{N}A]$. It is well known that Abelian semigroup algebras are local models for toric varieties.

\begin{deff}
We denote the semigroup above by $\Sigma$. It is an Abelian semigroup, which is generated by the set $A$ and $0$.  We denote the semigroup algebra by $S_{\Sigma}=k[t^a:a \in A]$ and the corresponding toric variety by $X_{\Sigma}=\mathbf{spec}[S_{\Sigma}]$. $\phi_1$ induces a grading on $\Sigma$. We denote the projective toric variety $\mathbf{Proj} (S_{\Sigma},\phi_1)$ by $\mathbb{P}_{\Sigma}$ and the corresponding line bundle by $\mathcal{O}(1)$.
\end{deff}
\begin{rem}
Toric varieties defined by semigroup algebras are not necessarily normal, while all toric varieties defined by a rational polyhedral fan are normal. It turns out that a spectrum of an Abelian semigroup ring is normal iff the semigroup is saturated, i.e. iff all $\mathbb{Z}^n$ points in the real cone generated by the semigroup are in the semigroup. (see \cite{Hoch})

\end{rem}
Note that the $k$-vector space of global sections of $\mathcal{O}(1)$ is canonically isomorphic to $V$, since degree one elements in the semigroup are a basis for $V$. We consider $f$ as the universal section of $\mathcal{O}(1)$. Let $Y=\mathbb{V}(f)$ be the codimension one subvariety of the zeros of $f$ in $V\times \mathbb{P}_{\Sigma}$ and let $U$ be the complement.

\begin{lem}
We have isomorphisms
\begin{align*}
\textbf{spec} (\mathfrak{R})=& V \times X_{\Sigma}\\
\textbf{spec} (\proj)=& U = V \times \mathbb{P}_{\Sigma} \setminus Y
\end{align*}

\begin{proof}
The first part follows from the definition of $X_{\Sigma}$ as the spectrum of $k[t^a:a\in A]$. For the second part, note that, since  $V\times \mathbb{P}_{\Sigma} \setminus Y$ is an affine chart in $V\times\mathbb{P}_{\Sigma}$, its coordinate ring can be described by degree zero elements in $\mathfrak{R}[1/f]$. This agrees with the definition of the $\proj$ as in (\ref{betazero}).
\end{proof}
\end{lem}
We state some well known facts from the theory of toric varieties. \\
\begin{deff}\label{polytope}
Let $P_A$ (respectively, $\bar{P}_A$) be the convex hull of the points in $A$ (respectively, $A\cup \{0\}$), 
as a subset of $\mathbb{R}^n$. $P_A$ has the structure of $n-1$ dimensional polytope and $\bar{P}_A$ has the structure of $n$ dimensional polytope. 
\end{deff}
By a $k$-dimensional face of a $m$-dimensional polytope $P$ we mean points in $P$ that lie on a hyperplane, with $k$ dimensional span (as an affine subspace) such that all points of $P$ are on the same side of the hyperplane. There is exactly one $m$-dimensional face. Codimension one faces are called {\em facets}. 
\begin{lem}
Faces of the $P_A$ (respectively, $\bar{P}_A$) are in one to one correspondence with torus orbits in 
$\mathbb{P}_{\Sigma}$ (respectively, $X_{\Sigma}$).
\end{lem}

$P_A$ is a subset of $\mathbb{R}^n$, where the the first coordinate is $1$. Thus, we can find hyperplanes defining faces that pass through the origin. Such a hyperplane can be defined as the set of $x\in \mathbb{R}^n$ with $\langle w, x \rangle=0$, for a vector $w$ in $\check{\mathbb{R}^{n}}$. In the case where points are all integer points, we can choose $w$ to have integer coefficients as well. The following definition gives an algebraic description of the torus orbits.

\begin{deff}\label{boundary}
Let $w\in \mathbb{Z}^n$ be a lattice point with the property $\langle w,a\rangle  \geq 0$ for all $a \in \Sigma$  and denote by $\Sigma^c_w$ the set of elements of $\Sigma$ that have nonzero inner product with $w$, i.e. $\langle w,a\rangle>0$, for $a \in \Sigma^c_w$. Moreover, denote by $\Sigma_w$ the set of elements with zero inner product , i.e. $\langle w,a \rangle =0$ for $a \in \Sigma_w$. Note that the sub $k$-vector space of $S_\Sigma$ generated by $\Sigma^c_w$ is a graded ideal in $S_\Sigma$. We denote this ideal by $I_w$. The quotient ring, which we denote by $S_w$, is isomorphic to $k[\Sigma_w]$. We denote by 
$\mathbb{P}_w=\mathbf{Proj} (S_w,\Phi_1)$ the corresponding projective toric variety.
\end{deff}

The projective toric varieties defined above are not necessarily smooth. Assuming that the toric variety $X$ is smooth, we have the following exact sequence of coherent sheaves.
$$
0 \rightarrow \Omega^1_X \rightarrow \Omega^1(log D) \rightarrow \bigoplus_{\{Facet\}} \mathcal{O}_{\mathbb{P}_w}\rightarrow 0
$$
where $D$ is the union of codimension one orbits and the second map is the residue map.\\

For a non-smooth toric variety $\Omega^1(log D)$ is well defined. See \cite{ishida}. Note that the complement of $D$ is a $n-1$ torus. By definition its coordinate ring is the ring of degree zero elements in the graded ring $k[\Sigma-\Sigma]\simeq k[\mathbb{Z}^n]$. Here by $\Sigma-\Sigma$ ,we mean differences of the elements of $\Sigma$. It is equal to $\mathbb{Z}^n$ since $\Sigma$ generates $\mathbb{Z}^n$ as a $\mathbb{Z}$-module. Thus, it is isomorphic to $k[\mathbb{Z}^n/\mathbb{Z}]$. We denote $\mathbb{Z}^n/\mathbb{Z}$ by $\tilde{M}$. An element $\tilde{m} \in \tilde{M}$ gives us a rational function on $\mathbb{P}_{\Sigma}$, then $dlog(\tilde{m})$ is an element of $\Omega_X^1(log D)$, which we denote by same notation $\tilde{m}$. It turns out that $\Omega^1(log D)$ is always free 
and we have an isomorphism of coherent sheaves 
$$\Omega^1(log D) \simeq \tilde{M} \otimes_{\mathbb{Z}} \mathcal{O}_{\mathbb{P}_{\Sigma}}. $$ 

We define the free sheaf of algebraic differential forms with logarithmic poles along $D$ to be 
$$
\bigoplus_{i=0}^{n-1} \Omega^i(log D) \simeq \bigoplus_{i=0}^{n-1} \bigwedge\nolimits^i \tilde{M} \otimes_{\mathbb{Z}} \mathcal{O}_{\mathbb{P}_{\Sigma}}\, .
$$

We have a diagram of varieties
\begin{equation}
\label{diagram}
\begin{tikzcd}
V\times \mathbb{T}\setminus Y_0 \arrow[hook]{r}{j}\arrow{d}{p} &U \arrow[hook]{r} \arrow{d}{p}  &V\times \mathbb{P}_{\Sigma} \arrow{d}{p}\\
V \arrow{r}{id} &V \arrow{r}{id}  &V
\end{tikzcd}
\end{equation}
where $Y_0=Y\cap V\times \mathbb{T}$ and $\mathbb{T}= \mathbb{P}_{\Sigma}\setminus D$. We consider the sheaf of relative differential forms with logarithmic poles along $D$ on $V\times \mathbb{P}_{\Sigma}$. Since $p$ is the projection on the first factor, we have
$$
\Omega^{\bullet}_{V\times \mathbb{P}_{\Sigma}/ V}(log D) = \bigwedge\nolimits^{\bullet} \tilde{M} \otimes_{\mathbb{Z}} \mathcal{O}_{V\times \mathbb{P}_{\Sigma}} \, .
$$
Sections of this sheaf on $U$ are
$$
 \Omega^{\bullet}_{V\times \mathbb{P}_{\Sigma}/ V}(log D)(U)=\bigwedge\nolimits^{\bullet} \tilde{M} \otimes_{\mathbb{Z}} \proj \, .
$$
This graded sheaf comes with a differential 
$$
d:\Omega^{\bullet}_{V\times \mathbb{P}_{\Sigma}/ V}(log D)\rightarrow \Omega^{\bullet+1}_{V\times \mathbb{P}_{\Sigma}/ V}(log D), 
$$
which makes it a complex of sheaves. 
For each element $a\in \Sigma\subset M$, there is a corresponding $\tilde{a} \in \tilde{M}$. Note that the restriction of $d$ to $U$ acts as
$$d(m\otimes (t^I/f^{\phi_1(I)})) = (\tilde{I}\wedge m) \otimes (t^I/f^{\phi_1(I)}) - \phi_1(I)\sum_{a \in A} \left( (\tilde{a} \wedge m) \otimes (p_a t^{I+a}/f^{\phi_1(I)+1})\right) \, .  $$
Note that this is well defined since we have
\begin{align*}
d(m\otimes (f t^{I}/f^{\phi_1(I)+1}))&=
d(m\otimes (\sum_{a\in A} p_a t^{a+I}/f^{\phi_1(I)+1}))\\ &=\sum_{a\in A}((\tilde{a}+\tilde{I})\wedge m) \otimes (p_a t^{a+I}/f^{\phi_1(I)+1})\\ &- \sum_{a\in A}p_a (1+\phi_1(I))\sum_{b \in A} \left( (\tilde{b} \wedge m) \otimes (p_b t^{I+a+b}/f^{\phi_1(I)+2})\right)\\
&=(\tilde{I}\wedge m) \otimes (\sum_{a\in A} p_a t^a t^{I}/f^{\phi_1(I)+1})+\sum_{a\in A}(\tilde{a}\wedge m) \otimes (p_a t^{a+I}/f^{\phi_1(I)+1}) \\
&- (1+\phi_1(I))\sum_{b \in A} \left( (\tilde{b} \wedge m) \otimes (\sum_{a\in A}p_a t^a p_b t^{I+b}/f^{\phi_1(I)+2})\right) \\
&=(\tilde{I}\wedge m) \otimes (t^{I}/f^{\phi_1(I)})- \phi_1(I)\sum_{b \in A} \left( (\tilde{b} \wedge m) \otimes (p_b t^{I+b}/f^{\phi_1(I)+1})\right) \\
&=d(m\otimes (t^{I}/f^{\phi_1(I)})) \, .
\end{align*}
Moreover, it is not hard to check that $d^2$ is zero. Note that this differential is the standard definition of $d$ coming from derivation below on rational functions 
$$\left(t_i \frac{\partial}{\partial t_i} \right) \frac{t^J}{f^{\phi_1(J)}} =\frac{t_i \frac{\partial}{\partial t_i} t^J }{f^{\phi_1(J)}} - \frac{\phi_1(J)t_i \frac{\partial f}{\partial t_i} t^J }{f^{\phi_1(J)+1}}  \, . $$
We can change the differential to
$$\nabla_{\beta} =d+\beta_1 df/f \wedge- \sum_{i=2}^n \beta_i \frac{dt_i}{t_i}\wedge = d + \beta_1 \left(\sum_{a \in A} \tilde{a} \otimes p_a  t^a / f\right) \wedge -\sum_{i=2}^n \beta_i \frac{dt_i}{t_i}\wedge , $$
and we still have a complex, since we have 
\begin{equation}\label{diff}
\nabla_{\beta} = \frac{f^{\beta_1}}{t_2^{\beta_2}\cdots t_n^{\beta_n}}\circ d \circ \frac{t_2^{\beta_2}\cdots t_n^{\beta_n}}{f^{\beta_1}} \, .
\end{equation}

\begin{prop}\label{cohomology}
Consider the complex $(\Omega^{\bullet}_{V\times \mathbb{P}_{\Sigma}/ V}(log D)(U),\nabla_{\beta})$, whose terms are free \proj-modules 
and with $\mathcal{O}_V$ linear differential $\nabla_{\beta}$. The $n-1$ hyper-cohomology 
of this complex is isomorphic to $H_A(\beta)$, i.e.
$$
R^{n-1} p _* (\Omega^{\bullet}_{V\times \mathbb{P}_{\Sigma}/V}(log D)|_U,\nabla_{\beta})= H_A(\beta)\, .
$$
\end{prop}
By $R^{n-1}p_*$ we mean the $(n-1)$-th derived functor of $p_*$. Everything is in Zariski topology.

\begin{proof}
Since $p$ is affine, hyper-cohomology agrees with cohomology of the complex. Choose the standard basis $e_2,e_3,\cdots,e_n$ for $\tilde{M}$, which gives us coordinates $t_2,\cdots,t_n$ on $\mathbb{T}$. We have 
$$\bigwedge\nolimits^{n-2}\tilde{M} \otimes_{\mathbb{Z}} \proj=\oplus_{i=2}^n  \ \frac{dt_2}{t_2}\cdots \hat{\frac{dt_i}{t_i}}\cdots \frac{dt_n}{t_n}\otimes \proj$$
and with $\phi_i$ as Definition \ref{phi}, we have
\begin{align*}
&\left(d+\beta_1 df/f\wedge-\sum_{i=2}^n \beta_i \frac{dt_i}{t_i}\wedge\right)\frac{dt_2}{t_2}\cdots \hat{\frac{dt_i}{t_i}}\cdots \frac{dt_n}{t_n}\otimes t^I/f^{\phi_1(I)} \\&=
(\tilde{I}\wedge \frac{dt_2}{t_2}\cdots \hat{\frac{dt_i}{t_i}}\cdots \frac{dt_n}{t_n}) \otimes (t^I/f^{\phi_1(I)}) \\&- \phi_1(I)\sum_{a \in A} \left( (\tilde{a} \wedge \frac{dt_2}{t_2}\cdots \hat{\frac{dt_i}{t_i}}\cdots \frac{dt_n}{t_n}) \otimes (p_a t^{I+a}/f^{\phi_1(I)+1})\right)\\
&+ \beta_1 \left(\sum_{a \in A} \tilde{a} \otimes p_a t^a / f\right) \wedge  \frac{dt_2}{t_2}\cdots \hat{\frac{dt_i}{t_i}}\cdots \frac{dt_n}{t_n}\otimes t^I/f^{\phi_1(I)} \\
&-\frac{dt_2}{t_2}\cdots\frac{dt_n}{t_n}\beta_i \otimes t^I/f^{\phi_1(I)}\\
&=  \frac{dt_2}{t_2}\cdots\frac{dt_n}{t_n}\otimes\left(-\beta_i+ \phi_i(I)-(\phi_1(I)-\beta_1)\sum_{a \in A} \phi_i(\tilde{a})p_a t^a/f \right)t^I/f^{\phi_1(I)}\\
&=  \frac{dt_2}{t_2}\cdots\frac{dt_n}{t_n}\otimes \left( t_i \frac{\partial}{\partial t_i} -\beta_i -t_i\frac{\partial f}{\partial t_i} (-\beta_1+t_1\frac{\partial}{\partial t_1})\right)t^I/f^{\phi_1(I)} \, .
\end{align*}
Identifying top forms with \proj , we see that the top cohomology is exactly $H_A(\beta)$ by Proposition \ref{dwork}.
\end{proof}
\begin{deff}
Given an element $w$ of the dual of $\tilde{M}$ we have the contraction map 
$$\iota_w: \bigwedge\nolimits^{\bullet} \tilde{M}\rightarrow \bigwedge\nolimits^{\bullet-1}\tilde{M}
$$
$$
\iota_w(a_1\wedge \cdots \wedge a_k)=\sum_{i=1}^k (-1)^i \langle w , a_i \rangle a_1\wedge\cdots \wedge \hat{a_i} \wedge \cdots \wedge a_k \,  .
$$
We extend it \proj-linearly to $\proj \otimes \bigwedge\nolimits^{\bullet} \tilde{M}$.
\end{deff}
\begin{deff}\label{seminon}
	Let $w_1,\dots,w_m$ be elements of $\mathbb{Z}^n$ defining the facets of $P_A$. An element $\beta \in \mathbb{C}^{n}$ is called non-resonant (respectively, semi non-resonant) if $\langle w_i,-\beta+\mathbb{Z}^n \rangle\neq 0$ (respectively, if $\langle w_i, -\beta + \Sigma\rangle \neq 0$), for all $i$. 

\end{deff}

\begin{prop}\label{w}
Let $w$ correspond to a face of $P_A$, the polytope defined in Definition \ref{polytope}, and let $I_w$, $\mathbb{P}_w$, and $S_w$ be as in Definition \ref{boundary}. Let $i_w: U_w \rightarrow U$ be the fiber product 
$$
\begin{tikzcd}
U_w \arrow{r}{i_w} \arrow{d}{p} &U \arrow{d}\\
{V\times \mathbb{P}_{w}} \arrow{r} & V \times \mathbb{P}_{\Sigma},
\end{tikzcd}
$$ 
which is the inclusion  into $U$ of the intersection of the boundary components corresponding 
to $w$ and $U$. This inclusion is given by the ideal sheaf $I_w$. Furthermore, assume $\beta$ is semi non-resonant. Then the inclusion 
$$
(I_w^n \Omega^{\bullet}_{V\times \mathbb{P}_{\Sigma}/V}(log D)|_U, \nabla_{\beta})\xrightarrow{ q.i.s}   (\Omega^{\bullet}_{V\times \mathbb{P}_{\Sigma}/V}(log D)|_U, \nabla_{\beta})
$$
 is a quasi-isomorphism.
\end{prop}
\begin{proof}
First note that $U_w$ is affine. By induction it is enough to show that 
$$
(I_w^n \Omega^{\bullet}_{V\times \mathbb{P}_{\Sigma}/V}(log D)|_U, \nabla_{\beta})\hookrightarrow   (I_w^{n-1}\Omega^{\bullet}_{V\times \mathbb{P}_{\Sigma}/V}(log D)|_U, \nabla_{\beta})
$$
is a quasi-isomorphism, or equivalently that the cokernel 
$$
(\frac{I_w^{n-1} \Omega^{\bullet}_{V\times \mathbb{P}_{\Sigma}/V}(log D)|_U}{ I_w^n \Omega^{\bullet}_{V\times \mathbb{P}_{\Sigma}/V}(log D)|_U}, \nabla_{\beta})
$$
is quasi-isomorphic to zero.\\

Let $t^J \otimes \alpha$ be a section on $I_w^{n-1} \Omega^{\bullet}_{V\times \mathbb{P}_{\Sigma}/V}(log D)|_U$. We show that a multiple of $\iota_{\tilde{w}} t^J \otimes \alpha$ is $\nabla_{\beta}$-primitive of $t^J\otimes \alpha$, where $\tilde{w}=(w_2,\cdots,w_n)$. We have
\begin{align*}
\nabla_{\beta} \iota_{\tilde{w}} t^J \otimes \alpha &= t^J\otimes J \wedge \iota_{\tilde{w}} \alpha + (-\phi_1(J)+\beta_1) \sum_{a\in A} p_a t^{a+J} \otimes a \wedge \iota_{\tilde{w}} \alpha - \sum \beta_i t^J \otimes e_i \wedge \iota_{\tilde{w}} \alpha\\
&= -t^J \otimes \iota_{\tilde{w}} J\wedge \alpha - (-\phi_1(J)+\beta_1)\sum_{a\in A} p_a t^{a+J} \otimes \iota_{\tilde{w}} a\wedge \alpha + \sum \beta_i t^J \otimes \iota_{\tilde{w}} e_i \wedge \alpha\\
&+  t^J \otimes (\iota_{\tilde{w}} J)\wedge \alpha + (-\phi_1(J)+\beta_1)\sum_{a\in A} p_a t^{a+J} \otimes (\iota_{\tilde{w}} a)\wedge \alpha - \sum \beta_i t^J \otimes (\iota_{\tilde{w}} e_i) \wedge \alpha\\
&= -\iota_{\tilde{w}} \nabla_{\beta} t^J\otimes \alpha + \left(t^J \langle (0,\tilde{w}), J \rangle+(-\phi_1(J)+\beta_1)\sum_{a \in A}P_a t^{a+J} \langle (0,\tilde{w}) , a \rangle - t^J\sum \beta_i \phi_i(w) \right) \otimes \alpha \, .
\end{align*}
Note that 
\begin{align*}
\langle (0,\tilde{w}), J \rangle =&\langle w , J \rangle -  w_1 \phi_1(J)\\
\langle (0,\tilde{w}), a \rangle =&\langle w , a \rangle -  w_1 \\
\sum_{a \in A} P_a t^{a+J} =& t^J  \, .
\end{align*}
Using these equalities we can rewrite the equation as 
$$
\nabla_{\beta} \iota_{\tilde{w}} t^J \otimes \alpha=-\iota_{\tilde{w}} \nabla_{\beta} t^J\otimes \alpha + \left(t^J \langle w, J \rangle+(-\phi_1(J)+\beta_1)\sum_{a \in A}P_a t^{a+J} \langle w , a \rangle - t^J \langle w , \beta \rangle \right) \otimes \alpha \, .
$$

Note that $\langle w , J \rangle \in \mathbb{Z}_{\geq 0}$,  and terms that appear in the first sum are all zero, since if $a \in S_w$ we have $\langle w , a \rangle = 0$ and if $a \in S_w^c$ we have $t^{a+J} \in I_w^n \Omega^{\bullet}_{V\times \mathbb{P}_{\Sigma}/V}(log D)|_U$. If $t^J \otimes \alpha$ is closed, we see that 
$$\langle w , J \rangle t^J \otimes \alpha-t^J\sum \beta_i \langle w , e_i \rangle \otimes \alpha= (\langle w , J \rangle -\langle w, \beta \rangle) t^J \otimes \alpha= c\ t^J \otimes \alpha$$ 
is exact, where $c$ is a nonzero constant, by the semi non-resonance assumption.
\end{proof}
In \cite{andre}, authors construct a category of complexes of sheaves. De Rham complexes live in this category. Given an open embedding $j:X\to Y$ with $X$ smooth and a $D$-module $M$ on $X$, they define a $Rj_!$ functor on this category. If $Y$ is smooth they show:
$$DR(j_! M)= Rj_!DR(M)$$
where $j_!$ is the left adjoint of $j^!$ operator on holonomic $D$-modules as in \cite{borel}. See appendix D of \cite{andre}.

\begin{cor}\label{shreik}
	Let $j:V\times \mathbb{T} \setminus Y_0 \rightarrow U$ be the inclusion of the complement of the torus boundary as in (\ref{diagram}). Assuming that $\beta$ is semi non-resonant, we have 
	$$ Rj_{!} j^*(\Omega^{\bullet}_{V\times \mathbb{P}_{\Sigma}/V}(log D)|_U, \nabla_{\beta})
\xrightarrow{ q.i.s}   (\Omega^{\bullet}_{V\times \mathbb{P}_{\Sigma}/V}(log D)|_U, \nabla_{\beta}) \, .
$$
Here we think of complexes as objects of the derived category of sheaves of Abelian groups. The functor
$j^*$ is the pull back functor and $Rj_!$ is as in Definition D.2.14 of \cite{andre}.
\end{cor}
\begin{proof}
	To compute $Rj_!$ we need to take an extension of $j^*(\Omega^{\bullet}_{V\times \mathbb{P}_{\Sigma}/V}(log D)|_U, \nabla_{\beta})$ to $U$ and take the limit by the ideal defining the complement. We already have an extension and we have shown that the powers of the ideal defining the boundary do not change the cohomology. Since both varieties are affine, the limit is $(\Omega^{\bullet}_{V\times \mathbb{P}_{\Sigma}/V}(log D)|_U, \nabla_{\beta})$.
\end{proof}
\begin{lem}\label{derahm}
Consider $\mathcal{O}_{V\times \mathbb{T}}$ as a D-module and let $DR_{V\times \mathbb{T}/V}(\mathcal{O}_{V\times \mathbb{T}})$ be its relative de Rham complex. Assume $\beta$ has integer coefficients. Then we have a quasi-isomorphism 
$$ j^*(\Omega^{\bullet}_{V\times \mathbb{P}_{\Sigma}/V}(log D)|_U, \nabla_{\beta})
\xrightarrow{ q.i.s}   (\Omega^{\bullet}_{V\times \mathbb{T}/V}|_{V\times \mathbb{T}\setminus Y_0}, d)= DR_{V\times \mathbb{T}\setminus Y_0/V}(\mathcal{O}_{V\times \mathbb{T}\setminus Y_0})\,  . $$
\end{lem}
\begin{proof}
	This simply follows from the fact that the $t_i$'s are invertible on $\mathbb{T}$. The isomorphism is given by twisting the differential by $$\frac{f^{\beta_1}}{t_2^{\beta_2}\dots t_n^{\beta_n}}$$ and its inverse as in (\ref{diff}). Note that, since $A$ generates $\mathbb{Z}^n$ as a $\mathbb{Z}$-module, multiplication by the rational function above is an isomorphism. 
\end{proof}
\begin{theorem}
	Assume $\beta$ has integer coefficients, is semi non-resonant and $\beta_1$ is negative. Then we have 
	$$R^{n-1}p_*(Rj_! DR_{ V\times \mathbb{T}\setminus Y_0/ V}(\mathcal{O}_{V\times \mathbb{T}\setminus Y_0}))	=H_A(\beta)\, . $$
\end{theorem}
\begin{proof}
	Proposition \ref{cohomology} and Lemma \ref{derahm} together with Corollary \ref{shreik} imply this.
\end{proof}
For the following theorem, we assume that the toric variety is normal, which is equivalent to 
$\Sigma$ being saturated.

\begin{theorem}\label{motive}
	Assume $\beta$ has integer coefficients, is semi non-resonant and $\beta_1$ is negative. Assume the semigroup $\Sigma$ is saturated. Let $U_v$ be the fiber of $p$ over $v\in V$ and let $D$ be the boundary divisor, i.e. the complement of $V\times \mathbb{T}$. We have 
$$
H^0(Sol(H_A(\beta)))_v :=Hom_{D_V}(H_A(\beta), \mathcal{O}^{an}_{V,v})= H_{n-1}(U_v,U_v\cap D)\, .
$$
\end{theorem}
\begin{proof}
If $\mathbb{P}_{\Sigma}$ was smooth, we could consider $D$-modules on $U$. Note that taking the relative de Rham complex commutes with $j_!$ by , i.e. we have 
$$Rj_!DR_{V\times \mathbb{T}\setminus Y_0/V}(\mathcal{O}_{V\times \mathbb{T}\setminus Y_0})=DR_{U/V}(j_!\mathcal{O}_{V\times \mathbb{T}\setminus Y_0})$$
We have $p_+$ functor of \cite{borel}. By definition of $p_+$ for projections,  
$R^{n-1}p_* \circ DR_{U/V}= H^0 p_+$. This implies
$$
H_A(\beta) = H^0 p_+ (j_! \mathcal{O}_{V\times \mathbb{T}\setminus Y_0})
$$
as quasi-coherent sheaves. The fact that connections agree follows from direct computation. The rest of the proof is the same as in \cite{zhu}. The idea is to use the Riemann-Hilbert correspondence and the fact that $Sol \circ j_! = j_* \circ Sol$ from section $15$ of \cite{borel} and the sheaf theoretic definition of relative homology from lemma $3.4$ of \cite{zhu}. Furthermore, the isomorphism is given by the cycle to period map defined in \cite{zhu2}.\\

$X=\mathbb{P}_{\Sigma}$ is normal since $\Sigma$ is saturated. By \cite{ishida}, there exists an equivariant resolution of singularities $g:X'\rightarrow X$, such that $X'$ is smooth and $Rg_* \mathcal{O}_{X'}= \mathcal{O}_X$. We have a diagram of fiber products of the form 
$$\begin{tikzcd}
	V\times \mathbb{T} \setminus Y_0 \arrow[hook]{r}{j'}\arrow{d}{id} &U' \arrow[hook]{r} \arrow{d}{g'}&V\times X' \arrow{d}{id \times g} \\
	V\times \mathbb{T}\setminus Y_0  \arrow[hook]{r}{j} &U \arrow[hook]{r} &V\times X
\end{tikzcd}
$$

By the same computations in Proposition \ref{w}, one can show that the de Rham complex for $j'_! \mathcal{O}$ can be computed using the de Rham complex with logarithmic poles with twisted differential. The twisted de Rham complex is again a complex of free $O_{X'}$ modules. Since $Rg'_*\mathcal{O}_{U'}= \mathcal{O}_U$ we have 
$$Rg'_* DR_{U'/V}(j'_!\mathcal{O}_{V\times \mathbb{T}\setminus Y_0})=Rg'_*(\Omega^{\bullet}_{V\times X'/V}log(D)|_{U'},\nabla_{\beta})= (\Omega^{\bullet}_{V\times X/V}log(D)|_{U},\nabla_{\beta})\, .$$
Thus, we have
\begin{align*}
H_A(\beta)&= R^{n-1}p_*  (\Omega^{\bullet}_{V\times X/V}log(D)|_{U},\nabla_{\beta})
\\&= R^{n-1}p_* Rg'_* DR_{U'/V}(j'_! \mathcal{O}_{V\times \mathbb{T}\setminus Y_0})\\
&= R^{n-1}(p\circ g')_* DR_{U'/V}(j'_! \mathcal{O}_{V\times \mathbb{T}\setminus Y_0})= H^0 (p\circ g')_+(j'_! \mathcal{O}_{V\times \mathbb{T}\setminus Y_0})\, .
\end{align*}
Applying the $Sol$ functor we get
$$Sol(j'_!\mathcal{O}_{V\times \mathbb{T}\setminus Y_0})=Rj'_*\mathbb{C}_{V\times \mathbb{T}\setminus Y_0}\, . $$ 
Moreover, using $Sol(p\circ g')_+= R(p\circ g')_! Sol[n-1]$ from section 14 of \cite{borel} , we have
$$Sol(H_A(\beta))= \leftidx{^p}R^0 R(p\circ g')_! Rj'_* \mathbb{C}_{V\times \mathbb{T}\setminus Y_0}[n-1]=\leftidx{^p}R^0 Rp_! \circ Rg'_!\circ Rj'_* \mathbb{C}_{V\times \mathbb{T}\setminus Y_0}[n-1] \, .$$
However, $g'$ is proper, therefore $Rg'_!=Rg'_*$ and we have 
$$Rg'_*\circ Rj'_*= R(g'\circ j')_*= Rj_*  \, .$$
Moreover, we have $$Sol(H_A(\beta))=\leftidx{^p}R^0 Rp_!\circ Rj_* \mathbb{C}_{V\times \mathbb{T}\setminus Y_0}[n-1] \, .$$
Thus, the sheaf of classical solution is 
$$H^0(Sol(H_A(\beta)))= R^{n-1}p_! \circ Rj_*  \mathbb{C}_{V\times \mathbb{T}\setminus Y_0}\, .$$
Note that $p_!$ commutes with taking fiber, by the compact support base change theorem. We want to find the restriction of $Rj_*\mathbb{C}_{V\times \mathbb{T}\setminus Y_0}$ to the fiber $U_v$ of $p$ over a point $v\in V$.  Denote the inclusion of $V\times\mathbb{T}$ into $V \times \mathbb{P}_{\Sigma}$ by $\bar{j}$ and denote the map from $U$ to $V\times \mathbb{P}_{\Sigma}$ by $i$. We have a diagram of varieties.
\begin{equation*}
    \raisebox{-0.5\height}{\includegraphics{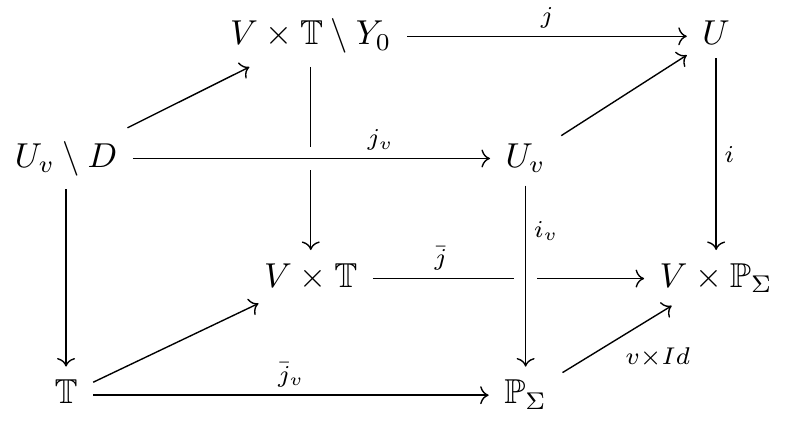}}
\end{equation*}
Since $Rj_*$ is local on the target, we have $Rj_* \mathbb{C}_{V\times\mathbb{T}\setminus Y_0} = i^*R\bar{j}_* \mathbb{C}_{V\times\mathbb{T}}$. Therefore 
$$Rj_* \mathbb{C}_{V\times\mathbb{T}\setminus Y_0}|_{U_v} = i^*R\bar{j}_* \mathbb{C}_{V\times\mathbb{T}}|_{U_v}= i_v^* (v\times Id)^* R\bar{j}_* \mathbb{C}_{V\times\mathbb{T}}\ .$$
Since $\bar{j}$ does not depend on the $V$, we have $(v\times Id)^* R\bar{j}_* \mathbb{C}_{V\times \mathbb{T}} = R\bar{j}_{v*} \mathbb{C}_\mathbb{T}$. From the square in the front, we see that $i^*_v R\bar{j}_{v*} \mathbb{C}_\mathbb{T}= Rj_{v*} \mathbb{C}_{U_v\setminus D}$. Thus, $j_*$ commutes with restriction to a fiber.
$$Rj_* \mathbb{C}_{V\times\mathbb{T}\setminus Y_0}|_{U_v} = Rj_{v*} \mathbb{C}_{U_v\setminus D}$$
By sheaf theoretic definition of relative homology we deduce
$$Hom_{D_V}(H_A(\beta), \mathcal{O}^{an}_{V,v})=  (R^{n-1}p_! \circ Rj_*  \mathbb{C}_{V\times \mathbb{T}\setminus Y_0})_v= H_{n-1}(U_v,U_v\cap D)\, . $$
\end{proof}

We showed that, for semi-nonresonant integer $\beta$, the sheaf of classical solutions to $H_A(\beta)$ is isomorphic to a relative homology. To find the isomorphism one needs to follow the proof of the Riemann-Hilbert correspondence which can be found in \cite{zhu2}. Assume that $v\in V$ is point and consider a relative chain $\delta_v \in H_{n-1}(U_v,U_v\cap D)$. We can extend this cycle to an analytic neighborhood of $v$ in $V$. Let $\phi(v)$ be the function defined in the neighborhood of $v$ by 
$$\phi(v)= \int_{\delta_v} \frac{f^{\beta_1}}{t^{\tilde{\beta}}}\frac{dt_2}{t_2}\dots\frac{dt_n}{t_n}\, . $$
Some computations similar to what we had in Proposition \ref{GKZ3} show that this function satisfies 
$H_A(\beta)$. In fact the isomorphism in the Riemann-Hilbert correspondence comes from this morphism. 
We do not use this fact in the general form, since in our case the cycle $\delta$ is always the positive real numbers. For more precise formulation see \cite{zhu2}. \\

\section{\textbf{Regularization}}\label{integralsection}
To have a better notation, we change the dimension from $n$ to $n+1$. Assume $A\subset \mathbb{Z}^{n+1}$ such that all points have first coordinate equal to $1$. We denote the semigroup generated by $A$ by $\Sigma$ and we assume $\Sigma-\Sigma = \mathbb{Z}^{n+1}$. As before, let $P_A$ be the convex hull of points in $A$ and let $C_A$ be the cone consisting of rays originating from zero and passing through $P_A$. For $(\beta_0,\dots,\beta_n)=\beta\in \mathbb{C}^{n+1}$, we can consider the corresponding differential equations and integral forms of the solutions.\\

Assume all $P_a$'s are positive i.e. $v$ is a positive real point of $V$. We claim that the closure of $\mathbb{R}_+^n$ is a relative chain for the pair $(U_v, U_v \cap D)$. To check this, it is enough to show that $f$ does not vanish on the closure of $\mathbb{R}_+^n$. One can check that closure of $\mathbb{R}^n_+$ is homeomorphic to the polytope $P_A$ by moment map. The restriction of $f$ to each torus orbit corresponding to a face defined by $w$ is 
$$\sum_{\langle w,a \rangle= 0} P_a t^{\tilde{a}}\, . $$
Since $P_a$ are positive and $t^{\tilde{a}}$ are positive, we see that $f$ does not vanish. Assuming $\beta$ is semi non-resonant and integer, it follows that $\phi$ is well defined and the integral
$$\int_{\mathbb{R}_+^n} \frac{t^{\tilde{\alpha}}}{f^{\alpha_0}}\frac{dt_1}{t_1}\dots\frac{dt_n}{t_n}$$
is convergent, where $\alpha = - \beta$ and $\alpha=(\alpha_0,\tilde{\alpha})$. 
We prove this fact for more general values of $\beta$.

\begin{lem}
	Assume the $P_a$'s are positive real numbers indexed by $a\in A$. The integral 
\begin{equation}
  \mathcal{I}:=\int_{\mathbb{R}_{+}^{n}}\frac{t^{\tilde{\alpha}}}{\left( \sum_{a\in A}P_{a}t^{\tilde{a}} \right)^{\alpha_{0}}}\frac{dt_1}{t_1}\dots\frac{dt_n}{t_n}
  \label{convergence}
\end{equation}
converges for  $\alpha=(\alpha_0,\tilde{\alpha})$ iff $\Re(\alpha)$ is in the interior of $C_A$.
\end{lem}
\begin{proof}
	To show this, we reparametrize $\mathbb{R}_{+}^{n}$ by $\mathbb{R}^n$, with the map $t_i=\exp\left({x_{i}}\right)$, which gives 
\begin{equation}
  \int_{\mathbb{R}^{n}}\frac{\exp\langle \tilde{\alpha},x\rangle}{\left( \sum_{a\in A} P_a \exp\langle \tilde{a},x\rangle\right)^{\alpha_0}}dx_1 \dots dx_n \, .
  \label{hello}
\end{equation}
For any $a\in A$, we have
\begin{equation*}
|  \frac{\exp\langle \tilde{\alpha},x\rangle}{\left( \sum_{a\in A} P_a \exp\langle \tilde{a},x\rangle\right)^{\alpha_0}}| \leq \frac{\exp\langle \Re(\tilde{\alpha}),x\rangle}{\left(  P_a \exp\langle \tilde{a},x\rangle\right)^{\Re(\alpha_0)}}\leq C \exp \langle \Re(\tilde{\alpha}-\alpha_0 \tilde{a}),x\rangle ,
\end{equation*}
where $C$ is a constant that only depends on the $P_a$'s. 
Taking the minimum on $A$ we can rewrite the inequality as 
$$
\int_{\mathbb{R}^{n}}\frac{\exp\langle \tilde{\alpha},x\rangle}{\left( \sum_{a\in A} P_a \exp\langle \tilde{a},x\rangle\right)^{\alpha_0}}dx_1 \dots dx_n \leq C' \int_{\mathbb{R}^n} \exp\left( \min_a \langle \Re(\tilde{\alpha}-\alpha_0\tilde{a}),x\rangle \right) dx_1\dots dx_n \, .
$$
Note that $\Re(\alpha)$ is in the interior of the cone iff $\Re(\tilde{\alpha})$ is in the interior of the convex hull of $\{\Re(\alpha_0)\tilde{a}\}_{a\in A}$. For a fixed $x$, there exists at least one $a$ such that 
$$\langle \Re(\tilde{\alpha}-\alpha_0\tilde{a}),x \rangle =\langle \Re(\tilde{\alpha}) ,x\rangle- \langle \Re(\alpha_0) \tilde{a},x\rangle< 0, $$
otherwise all point of $\Re(\alpha_0) \tilde{a}$ would be in the half-space $\langle\  \cdot \ , x \rangle \geq \langle \tilde{\alpha},x\rangle$ and $\Re(\tilde{\alpha})$ would be on the boundary, which would contradict the fact that $\tilde{\alpha}$ is in the interior of the convex hull. To show that the integral converges, we use radial coordinates and we write 
\begin{equation*}
  \mathcal{I}\leq C' \int_{S^{n-1}\times \mathbb{R}_+} \exp \left(r\  \min_a \langle \Re(\tilde{\alpha}-\alpha_0\tilde{a}),\frac{x}{|x|}\rangle  \right) r^{n-1} dr\  d\Omega \, .
\end{equation*}
We showed that, for any $x$, $\min_a\langle \Re(\tilde{\alpha}-\alpha_0\tilde{a}),x\rangle$ is negative. Let $-\varepsilon$ be the supremum of this function on the sphere of radius one, which is negative by compactness of the sphere. Substituting, we get 
$$
\mathcal{I} \leq C' \int_{S^{n-1}\times \mathbb{R}_+}  r^{n-1}e^{-\varepsilon r}dr\ d\Omega \leq C' \int_{S^{n-1}} \frac{\Gamma(n) }{\varepsilon^{n}} d\Omega < \infty \, .
$$
\end{proof}
We want to find relations among integrals with different $\alpha$. Let $K(\alpha,P_a)$ be the integral
\begin{equation}\label{I_P}
K(\alpha, P_a)= \int_{\mathbb{R}_{+}^{n}}\frac{t^{\tilde{\alpha}}}{\left( \sum_{a\in A}P_{a}t^{\tilde{a}} \right)^{\alpha_{0}}}\frac{dt_1}{t_1}\dots\frac{dt_n}{t_n} \, .
\end{equation}
\begin{lem}\label{relations}
	Assume $\Re(\alpha)$ is in the interior of the cone. Let $w\in \mathbb{Z}^{n+1}$. We have the 
	equality 
$$
\langle w, \alpha \rangle K(\alpha,P_a)= \alpha_0 \sum_{a\in A}\langle w, a \rangle P_a K(\alpha+a,P_a) \, .
$$
\end{lem}
\begin{proof}
Both sides are linear in $w$, hence it is enough to check this for $w=e_i$. Note that, for $w=e_0$, this equality is trivial. For $i\neq 0$ and $w=e_i$, we use exponential change of variable as in equation (\ref{hello}). Consider the differential form on $\mathbb{R}^n$ given by 
$$\theta = \frac{\exp\langle \tilde{\alpha},x\rangle}{\left( \sum_{a\in A} P_a \exp\langle \tilde{a},x\rangle\right)^{\alpha_0}}dx_1 \dots \hat{dx_i}\dots dx_n \, .$$
A basic computation shows that we have
\begin{align*}
d\theta &= \left(\alpha_i \frac{\exp\langle \tilde{\alpha},x\rangle  }{\left( \sum_{a\in A} P_a \exp\langle \tilde{a},x\rangle\right)^{\alpha_0}} - \alpha_0\frac{\sum_{a\in A} P_a  a_i \exp \langle \tilde{\alpha}+a , x \rangle}{\left( \sum_{a\in A} P_a \exp\langle \tilde{a},x\rangle \right)^{\alpha_0+1}}\right) dx_1 \dots dx_n  \, .
\end{align*}
Assuming that $K(\alpha, P_a)$ converges implies that $\Re(\alpha)$ is the interior of the cone generated by $A$. Since $\Re(\alpha)$ is in the interior of the cone, $\Re(\alpha)+a$ is also in the interior of the cone and all terms in $d\theta$ have absolutely convergent integrals. Therefore, we can integrate $d\theta$ in the interior of ball of radius $r$ and take the limit as $r$ goes to infinity. By Stokes' theorem the integral of $d\theta$ over a ball of radius $r$ is equal to integral of $\theta$ on a sphere of radius $r$. By the same computations as in the previous lemma, we can see that $\theta$ has exponential decay, while the volume of the sphere grows polynomially. This implies that the limit is zero. Note that the integral of $d\theta$ is equal to 
$$\langle e_i, \alpha \rangle K(\alpha,P_a)- \alpha_0 \sum_{a\in A}\langle e_i, a \rangle P_a K(\alpha+a,
P_a),$$
which implies the statement.

\end{proof}
\begin{theorem}\label{poles}
	Let $K(\alpha,P_a)$ be as in (\ref{I_P}), defined for $\Re(\alpha)$ in the interior of the cone generated by $A$. Then $K(\alpha,P_a)$ has meromorphic continuation to $\mathbb{C}^{n+1}$, with poles along the $-\Sigma$ translates of the hyperplanes defining the facets of the convex hull of $A$, or equivalently 
the semi-resonant $-\alpha$'s. Furthermore, we have the identity in Lemma \ref{relations}.
\end{theorem}
\begin{proof}
	First note that the integral is absolutely convergent and that the integrand is analytic in $\alpha$, which implies that the integral is holomorphic for $\Re(\alpha)$ in the interior of the cone, i.e. for $\langle w_i , \Re(\alpha) \rangle> 0 $. We use induction to show the statement for the sets $\langle w_i , \Re(\alpha) \rangle> m_i$, which cover $\mathbb{C}^{n+1}$ for $m_i \in \mathbb{Z}$. The statement is true for $\vec{m}=(0,\dots,0)$. Assume it is true for $\vec{m}=(m_1,\dots,m_{n+1})$. By Lemma \ref{relations} we have 
	$$\langle w_i , \alpha \rangle K(\alpha,P_a) = \alpha_0 \sum_{a\in A} \langle w_i ,a\rangle P_a K(\alpha +a , P_a)\, .$$
	Note that $\langle w, a\rangle = 0$ for $a$ in the facet defined by $w_i$. Thus, $\langle w_i , \alpha \rangle K(\alpha, P_a)$ can be expressed as a linear combination of $K(b,P_a)$, where, for all $b$, we have $\langle w_i, \Re(b) \rangle \geq m_i +1$ and $\langle w_j , \Re(b) \rangle \geq m_j$ for $j\neq i$. We can define the integral for $\langle w_i , \Re(\alpha) \rangle > m_i- 1$ and $\langle w_j , \Re(\alpha) \rangle > m_j$ by dividing both sides by $\langle w_i , \alpha \rangle$. This means that there is pole when $\langle w_i, \alpha\rangle$ is zero. By repeating this operation, we get poles on the $-\Sigma$ translates of $\langle w_i , \alpha \rangle = 0$. Note that the equation above is analytic in $\alpha$ and is valid for $\Re(\alpha)$ in the interior of the cone, hence it is valid everywhere. 
\end{proof}
\begin{rem}\label{zero}
  Note that in the identity of Lemma \ref{relations}, the sum on the right is multiplied by $\alpha_0$. If we start from an $\alpha$ with $\alpha_0 \in \mathbb{Z}_{\leq 0}$, we eventually multiply be zero, since all points in the interior of the cone have $\alpha_0>0$. Thus, $K(\alpha,P_a)$ has degree one zero along hyperplanes $\alpha_0 \in \mathbb{Z}_{\leq 0}$.
\end{rem}

We know that, for positive real $P_a$'s,  if $K(\alpha,P_a)$ converges, then it is a solution to $H_A(\beta)$. We claim that this is true for the analytic continuation of $K(\alpha,P_a)$ as well. From the identity of Lemma \ref{relations}, we know that for $\alpha$ with real part in the interior of the positive cone, the left hand side is a solution to $H_A(-\alpha)$. Note that both differential equations and relations are analytic in $\alpha$ and $P_a$, therefore they are valid for the analytic continuation of $K(\alpha, P_a)$. As a result we have that $K(\alpha,P_a)$ is a solution to $H_A(-\alpha)$ when $-\alpha$ is non semi-resonant. For a resonant $-\alpha$ we can find a vector $\vec{u}\in \mathbb{C}^{n+1}$ such that $-(\alpha+\epsilon \vec{u})$ is semi non-resonant for small enough $\epsilon$. We can take the Laurent expansion in the $\vec{u}$ direction and
we obtain 
\begin{equation}\label{regular}
	K(\alpha+\epsilon\vec{u}, P_a)=\sum_{i=-k}^\infty \epsilon^i K_i^{\vec{u}}(\alpha, P_a) \, . 
\end{equation}
As in the previous section, we denote the set of integer relations among points of $A$ by $R$.
For $r \in R$, we have a corresponding box differential operator $\Box_r$ as in (\ref{box}) and 
\begin{equation}\label{Zw}
Z_w = \sum_{a\in A} \langle w, a \rangle P_a\frac{\partial}{\partial P_a}\, .
\end{equation}
Note that $Z_i$ in (\ref{zi}) is $Z_{e_i}$ in \eqref{Zw}. In this notation, a solution $\phi$ to $H_A(\beta)$ is equivalent to a function satisfying 
$$\Box_r \phi=0 \quad \quad Z_w \phi = \langle w, \beta \rangle \phi \, .$$
By the same computation of Theorem \ref{main}, we arrive to the following proposition.
\begin{prop}
For $w\in \mathbb{C}^{n+1}$, we have
$$\Box_r K_i^{\vec{u}}(\alpha,P_a)= 0 \quad \quad Z_w K_i^{\vec{u}}(\alpha,P_a) = \langle w, -\alpha \rangle K_i^{\vec{u}}(\alpha,P_a)+ \langle w, -\vec{u} \rangle K_{i-1}^{\vec{u}}(\alpha,P_a) \, .$$
\end{prop}
\begin{proof}
This follows from expanding both sides of (\ref{regular}).
\end{proof}
In particular we see that the lowest coefficient gives us a solution to $H_A(-\alpha)$.
\begin{rem}
	All the calculations we have done here can be done for any chain $\delta$ replacing the positive real points. In fact, Lemma \ref{relations} is valid for the integral over any chain, and the rest of the calculation is exactly the same. In this way we can construct a set of solutions for resonant $\beta$. 
\end{rem}

\section{\textbf{Amplitudes and Regularization}}\label{feynGKZ}

As we showed in Section \ref{sec1}, amplitudes satisfy certain differential equations. From a Feynman diagram, we constructed a subset of $\mathbb{Z}^{n+1}$ in the following way. For each monomial $t^S$ in the first Symanzik polynomial $\Psi_{\Gamma}$ of the graph, we consider the point $(1,\vec{S}) \in \mathbb{Z}^{n+1}$ and, for each monomial $t^T$ in $Q_{\Gamma}$, we consider the lattice point $(0,\vec{T})$. We denote this set of lattice points by $A$. In general $A$ does not generate the lattice $\mathbb{Z}^{n+1}$ but it generates a sublattice of dimension $n$. To see this, note that $(1,0,0,...0,-1,0,...,0)$, where $-1$ is in the $i$-th place, is in the sublattice generated by $A$. This term is $(1,\vec{S})-(0,\vec{T})$, where the monomial $T$ is the product of $m_i^2  t_i$ and $t^S$. On the other hand, all points of $A$ lie on the hyperplane $\sum_{i=0}^n a_i = \ell+1$, where $\ell$ is the number of loops in $\Gamma$. Thus, the sublattice generated by $A$ in $\mathbb{Z}^{n+1}$ is the set of lattice points $x$, such that $\ell+1 |\sum x_i$. Denote this sublattice by $L$ and let $r:L\rightarrow \mathbb{Z}^{n+1}$ be the function defined by $r_0(x)=(x_0 +x_1+...+x_n)/(\ell+1)$ and $r_i(x)=x_i$, for $1\leq i \leq n$. It is easy to see that $r$ is invertible and that the determinant is $1/(\ell+1)$, hence the image of the standard cube has volume $1$ and $r(A)$ generates the lattice. We replace $A$ by $r(A)$. Note that all points of $r(A)$ have first coordinate equal to $1$. By Proposition \ref{GKZ2} we know that $I(c_1, c_2 , P_a , \vec{v})$ satisfies $H_A(\beta)$, where 
$$\beta = ( (c_1-c_2)(\ell+1)+ \sum_i v_i +n , -1-v_1,\dots,-1-v_n)\, . $$
Thus, $I$ satisfies $H_{r(A)}(r(\beta))$ and we have
$$r(\beta) = (c_1-c_2, -1-v_1, \dots ,-1-v_2) \, .$$
For the rest of this section we replace $A$ by $r(A)$ and $\beta$ by $r(\beta)$. Note that, for the original amplitude $\mathcal{A}(\Gamma, P_a)$, we have $c_1 =0$, $c_2 = D/2$ and $v= \vec{0}$, hence the corresponding vector $\beta$ is
$$\beta=(-D/2,-1,\dots,-1) \, . $$
Assuming the normality condition and the semi non-resonant condition in Theorem \ref{motive}, 
all solutions of $H_A(\beta)$ come from integrals. Thus, $\mathcal{A}(\Gamma,P_a)$ is equal to 
$$ \int_{\delta} \frac{t_1\dots t_n}{(\sum_{a\in A}P_a t^{\tilde{a}})^{D/2}}\frac{dt_1}{t_1}\dots 
\frac{dt_n}{t_n},$$
for a relative chain $\delta$.
In fact, using the projective version of the integral, we can show that it is equal to $K(-\beta, P_a)$, up to  multiplication by a rational number.

\begin{lem} \label{feyn}
Feynman's parametric integral formula gives
	$$\frac{1}{A^{a}B^{b}} = \frac{\Gamma( a+b)}{\Gamma(a)\Gamma(b)}\int_{0}^{\infty} \frac{\lambda^{a-1}d\lambda}{\left[\lambda A + B\right]^{a+b}} \, .$$
\end{lem}
\begin{proof}
	See Chapter 8 of \cite{critical}.
\end{proof}

\begin{cor}
  Assume $K(\alpha, P_a)$ is convergent, i.e. $\alpha$ is in the interior of the cone generated by $A$. Then we have the equality 
  $$K(\alpha, P_a) =\frac{ \Gamma(-|\tilde{\alpha}|+\alpha_0(\ell+1))}{\Gamma(\alpha_0)}I(0,\alpha_0,P_a,\tilde{\alpha}-(1,\dots,1)) \, . $$
\end{cor}
\begin{proof}
By Lemma \ref{feyn} we have 
\begin{align*}
	\int_{\mathbb{R}_+^n} \frac{t^{\tilde{\alpha}}}{\left( Q_{\Gamma}+\Psi_{\Gamma} \right)^{\alpha_0}}=&\int_{\mathbb{R}_+ \times \Delta_{n-1}} \frac{r^{|\tilde{\alpha}|+n-1}t^{\tilde{\alpha}}}{\left( r^{\ell+1}Q_{\Gamma}+r^{\ell}\Psi_{\Gamma} \right)^{\alpha_0}}dr \Omega\\
	=& \int_{\mathbb{R}_+ \times \Delta_{n-1}} \frac{r^{|\tilde{\alpha}|+n-1-\alpha_0 \ell}t^{\tilde{\alpha}}}{\left( rQ_{\Gamma}+\Psi_{\Gamma} \right)^{\alpha_0}}dr \Omega\\
	=&\frac{ \Gamma(-n-|\tilde{\alpha}|+\alpha_0(\ell+1))\Gamma(|\tilde{\alpha}|+n-\alpha_0\ell)}{\Gamma(\alpha_0)}\int_{\Delta_{n-1}} \frac{Q_{\Gamma}^{-n-|\tilde{\alpha}|+\alpha_0\ell}}{\Psi_{\Gamma}^{-n-|\tilde{\alpha}|+\alpha_0 (\ell+1)}} t^{\tilde{\alpha}}\Omega\\
	=&\frac{ \Gamma(-n-|\tilde{\alpha}|+\alpha_0(\ell+1))\Gamma(|\tilde{\alpha}|+n-\alpha_0\ell)}{\Gamma(\alpha_0)}J(-n-|\tilde{\alpha}|+\alpha_0\ell,-n-|\tilde{\alpha}|+\alpha_0(\ell+1),P_a,\tilde{\alpha}) \, .
\end{align*}
By the definition of $K(\alpha,P_a)$ and the equality (\ref{ijrelation}), we have 
\begin{align*}
	K(\alpha,P_a)&=\frac{ \Gamma(-|\tilde{\alpha}|+\alpha_0(\ell+1))\Gamma(|\tilde{\alpha}|-\alpha_0\ell)}{\Gamma(\alpha_0)}J(-|\tilde{\alpha}|+\alpha_0\ell,-|\tilde{\alpha}|+\alpha_0(\ell+1),P_a,\tilde{\alpha}-(1,\dots,1))\\
	&=\frac{ \Gamma(-|\tilde{\alpha}|+\alpha_0(\ell+1))}{\Gamma(\alpha_0)}I(0,\alpha_0,P_a,\tilde{\alpha}-(1,\dots,1)) \, . 
\end{align*}
\end{proof}

By corollary above and the fact that $\mathcal{A}(\Gamma,P_a)=\pi^{D\ell/2}I(0,D/2, P_a,\vec{0})$, we can compute the amplitude from $K$ as
\begin{equation}\label{ka}
\mathcal{A}(\Gamma,P_a)= \pi^{D\ell/2}\frac{\Gamma(D/2) }{\Gamma(-n+D/2(\ell+1))}K((D/2,1,\dots,1),P_a)\, . 
\end{equation}
Note that the Gamma function never vanishes. As a result, when $K(\alpha,P_a)$ converges, we can define the amplitude. In fact, we can define the amplitude by this equation. For the rest of this section we study $K(\alpha,P_a)$, and the structure of its poles. To do that, we first find the defining inequalities for $P_A$. 
We need a few definitions and notation.

\begin{deff}
  For a graph $\Gamma$ with $n$ edges, we denote by $SP_{\Gamma}\subset \mathbb{R}^n$ the polytope constructed by the incidence vectors of the complements of the spanning trees of $\Gamma$, i.e. it is the convex hull of $\{\vec{S}\}$, where $S$ corresponds to a monomial $t^S$ in the first Symanzik polynomial of $\Gamma$. We denote by $P_{\Gamma}$ the polytope constructed from the terms in the first and second Symanzik polynomials (including mass terms). Note that we have $P_{\Gamma} = SP_{\Gamma}+ E_n$, where $E_n$ is the convex hull of $\{0,e_1,\dots,e_n\}$ and plus is the Minkowski sum.
For a subgraph $\gamma \subset \Gamma$, we have natural inclusions $P_{\gamma}\subset P_{\Gamma}$ and $SP_{\gamma}\subset SP_{\gamma}$.
\end{deff}

\begin{deff}
	
	For a subgraph $\gamma \subset \Gamma$, let $1_{\gamma}$ be the incidence vector of $\gamma$ in $\mathbb{Z}^{n}$, i.e. the vector where the coefficient corresponding to an edge $e\in \Gamma$ is $1$ if $e\in \gamma$ and is zero otherwise. We denote by $\ell_{\gamma}$ the number of loops in $\gamma$, i.e.~the dimension of the first homology of $\gamma$.
\end{deff}

\begin{lem}\label{mayer}
	Given two subgraph $\gamma_1$ and $\gamma_2$, we have
	$$\ell_{\gamma_1} + \ell_{\gamma_2} \leq \ell_{\gamma_1\cap \gamma_2} + \ell_{\gamma_1 \cup \gamma_2} \, . $$
\end{lem}
\begin{proof}
	By Mayer-Vietoris for the pair $\gamma_1$ and $\gamma_2$, we have the exact sequence 
	$$0=H_2(\gamma_1 \cup \gamma_2) \rightarrow H_1(\gamma_1\cap \gamma_2) \rightarrow H_1(\gamma_1)\oplus H_1(\gamma_2) \rightarrow ker \left(H_1(\gamma_1\cup \gamma_2) \rightarrow H_0(\gamma_1\cap \gamma_2)\right)\rightarrow 0\, . $$
	Counting dimensions implies the inequality.
\end{proof}
We define another polytope using inequalities and we will show it is the same polytope we considered before. Spanning tree polytopes ($SP_{\Gamma}$) and in general matroid polytopes have been studied in combinatorial optimization theory, see \cite{polytope1}, \cite{polytope2} and \cite{polytope3}. The polytope $P_{\Gamma}$ is similar to these polytopes and we can translate some of the results in combinatorial optimization to our setting. 
\begin{deff}
  For a graph $\Gamma$, let $P'_{\Gamma}$ be the subset of $\mathbb{R}^n$ defined by the inequalities
	$$\langle 1_{\gamma}, \tilde{x}\rangle \geq \ell_{\gamma} \quad \quad
	\langle 1_{\Gamma}, \tilde{x}\rangle \leq \ell_{\Gamma} + 1, $$
where the first inequality is valid for all subgraphs. 
\end{deff}
Assume that the graph $\gamma$ in the first inequality above is a single edge. In this case, the inequality implies that all coefficients are positive. The second inequality implies the sum of the coefficients is bounded. Thus, these equations define a bounded set and $P'_{\Gamma}$ is a polytope. Any face of the polytope $P= P'_{\Gamma}$ is defined by setting some of the inequalities to equalities. Assume we have a set of equalities $\langle 1_{\gamma}, \tilde{x}\rangle = \ell_{\gamma}$ for $\gamma\in \mathcal{F}$ where $\mathcal{F}$ is a family of subgraphs. 
\begin{lem}\label{family}
Given a point $\tilde{x}$ in $P$, let $\mathcal{F}$ be the set 
$$\mathcal{F} = \{\gamma\subset\Gamma: \langle 1_{\gamma}, \tilde{x}\rangle = \ell_{\gamma}\} \, .$$
Then $\mathcal{F}$ is closed under intersection and union of its elements. 
\end{lem}
\begin{proof}
We have 
$$\langle 1_{\gamma_1\cap \gamma_2} ,\tilde{x}\rangle \geq \ell_{\gamma_1\cap \gamma_2}$$
and 
$$\langle 1_{\gamma_1\cup \gamma_2} ,\tilde{x}\rangle \geq \ell_{\gamma_1\cup \gamma_2}, $$
and we have $1_{\gamma_1}+1_{\gamma_2}= 1_{\gamma_1\cap \gamma_2}+ 1_{\gamma_1 \cup \gamma_2}$. Combining this with the inequality of Lemma \ref{mayer}, we have
$$  \ell_{\gamma_1}+ \ell_{\gamma_2} =\langle 1_{\gamma_1},\tilde{x}\rangle+\langle 1_{\gamma_1},\tilde{x}\rangle \geq \ell_{\gamma_1\cap \gamma_2}+\ell_{\gamma_1\cup \gamma_2}\geq \ell_{\gamma_1}+ \ell_{\gamma_2}, $$
hence both inequalities are equalities. This shows that $\mathcal{F}$ is closed under 
intersection and union of its elements. 
\end{proof}
By the previous lemma, the defining equations of faces (coming from subgraphs) can be chosen to be closed under intersection and union. By a chain of subgraphs we mean a family $\mathcal{C}$ of subgraphs such that, for $\gamma_1, \gamma_2 \in \mathcal{C}$, we have either $\gamma_1 \subset \gamma_2$ or $\gamma_2 \subset \gamma_1$. 

\begin{lem}\label{chain}
	Let $\mathcal{F}$ be the set of equalities corresponding to subgraphs and let $P_{\mathcal{F}}$ be the corresponding face. Let $\mathcal{C}\subset \mathcal{F}$ be a maximal chain in $\mathcal{F}$. The family of linear equations $\{\langle 1_{\gamma}, \tilde{x}\rangle = \ell_{\gamma}: \gamma \in \mathcal{F}\}$ is equivalent to $\{\langle 1_{\gamma}, \tilde{x}\rangle = \ell_{\gamma}: \gamma \in \mathcal{C}\}$, i.e. we have $P_{\mathcal{F}}=P_{\mathcal{C}}$.
\end{lem}
\begin{proof}
Given a subgraph $\gamma$, by a chain violation we mean a subgraph $c\in \mathcal{C}$ such that neither $c\subset\gamma$ nor $\gamma\subset c$. Since $\mathcal{C}$ is maximal, for any subgraph $\gamma$ in $\mathcal{F}\setminus \mathcal{C}$ there exists $c\in \mathcal{C}$ such that it is a chain violation for $\gamma$  and $P_{\mathcal{C}\cup \{c\}}\neq P_{\mathcal{C}}$, otherwise the statement follows. Among all of the subgraphs choose the one with minimal number of chain violations. By Lemma \ref{family}, we have $c\cap \gamma \in \mathcal{F}$ and $c\cup \gamma \in \mathcal{F}$. For all $\tilde{x} \in P_{\mathcal{F}}$ we have 
$$\langle 1_{\gamma} , \tilde{x} \rangle = \ell_{\gamma} \quad \quad \langle 1_c , \tilde{x}\rangle = \ell_c\quad \quad \langle 1_{\gamma\cap c},\tilde{x}\rangle = \ell_{\gamma\cap c} \quad \quad \langle 1_{\gamma\cup c},\tilde{x}\rangle = \ell_{\gamma\cup c}\, . $$
These four linear equations are dependent, i.e. the sum of the first two is equal to the sum of the last two. The first equation is not satisfied by all points of $P_{\mathcal{C}}$, so we either have $P_{\mathcal{C}\cup\{\gamma \cap c\}} \neq P_{\mathcal{C}}$ or $P_{\mathcal{C}\cup\{\gamma \cup c\}} \neq P_{\mathcal{C}}$. Replacing $c$ by  $\gamma\cap c$ or $\gamma \cup c$ decreases the number of chain violations, which  contradicts the maximality condition. 
\end{proof}
\begin{prop}\label{polytopeineq}
  Two polytope are equal, i.e. $P_{\Gamma}= P'_{\Gamma}$. Furthermore let $C_{\Gamma}$ be the cone over $P_{\Gamma}$ in $\mathbb{R}^{n+1}$ and let   
$$\vec{\gamma} = (-\ell_{\gamma}, 1_{\gamma})$$
then $C_{\Gamma}$ is given by the equalities
$$\langle \vec{\gamma} , x \rangle \geq 0, $$
for all $\gamma \subset \Gamma$, and 
$$\langle \vec{\Gamma}- e_0 , x \rangle \leq 0 \, .$$
\end{prop}

\begin{proof}
  Note that the inequalities define a cone since $0$ satisfies all the equations. The intersection of this cone with the hyperplane $\langle x,e_0\rangle=1$ is $P'_{\Gamma}$, which is bounded. Since the intersection is bounded, the cone defined by the inequalities is the cone over the intersection. As a result it is enough to show that the intersection is equal to $P_{\Gamma}$. The first inequality for a subgraph is satisfied by all points of $A$, since an intersection of a spanning tree with a subgraph is a spanning forest of the subgraph, hence number of edges in its complement is greater than the number of loops. The second inequality is trivial. $P'_{\Gamma}$ contains all the extreme points of $P_{\Gamma}$, hence it contains $P_{\Gamma}$. First we find the integer points of $P'_{\Gamma}$. Any edge $e\in E$ is a subgraph and the corresponding first inequality implies $x_i\geq 0$. Thus, integer points of $P'_{\Gamma}$ have the form 
$$(a_1,\dots,a_n),$$ 
where $a_i \in \mathbb{Z}_{\geq 0}$. For any loop $\gamma \in \Gamma$ we have
$$\sum_{e\in \gamma} a_e \geq 1,$$ 
hence, for at least one $e\in \gamma$, we have $a_e \neq 0$. If we remove this edge from the graph, the remaining graph has $\ell_{\Gamma}-1$ loops and we can find another loop that does not have $e$ in it. 
This in turn implies that another $a_{e'}$ is non-zero. Iterating this procedure, we get at least $\ell_{\Gamma}$ many nonzero $a_i$'s, which are chosen from different loops. The complement of the edges corresponding to these $a_i$'s is a spanning tree. The second inequality implies $\sum_i a_i \leq \ell_{\Gamma}+1$, hence $(a_1,\dots,a_n)$ corresponds to the incidence vector of a complement of a spanning tree, or $e_i$ plus the incidence vector of a complement of a spanning tree. The integer points of $P'_{\Gamma}$ are exactly the integer points of $P_{\Gamma}$. To finish the proof we need to show that the extreme points of $P'_{\Gamma}$ 
are integers. \\

Let $\tilde{x}$ be an extreme point of $P'_{\Gamma}$. Then $\tilde{x}$ is the unique solution to a set of linear equations corresponding to defining equations of $P'_{\Gamma}$. There are two cases. The first case occurs  when the second equation is not used and $\tilde{x}$ is defined by $\langle 1_{\gamma} ,\tilde{x} \rangle = \ell_{\gamma} \text{ for } \gamma \in \mathcal{F}$. By Lemma \ref{chain}, we can replace $\mathcal{F}$ by a chain of subgraphs $\mathcal{C}$. Since the solution is unique, we need $n$ many equalities and we have  
$$\gamma_1 \subset \gamma_2 \subset \dots \subset \gamma_n = \Gamma \quad \quad \mathcal{C}= \{\gamma_1,\dots,\gamma_n\}. $$
Thus, $\gamma_i \setminus \gamma_{i-1}$ has only one edge. We denote this edge by $e_i$. 
By the equalities we see that 
$$x_i = \ell_{\gamma_i} - \ell_{\gamma_{i-1}} \in \{0,1\},$$ which implies we have an integer point. The set of extreme points we obtain in this way corresponds to a spanning trees. 
The second case is when we have the equality $\langle 1_{\Gamma}, \tilde{x}\rangle= \ell_{\Gamma}+1$. By the same argument, we see that $x_i \in \{0,1\}$, for all $i$ except one of them. Since these add up to an integer, the other one has to be an integer too. The set of extreme points we obtain in this way corresponds to the set of monomials in the second Symanzik polynomial.
\end{proof}
A subgraph is called 2-connected, if it is connected and remains connected after removing any vertex. Note that single edges are 2-connected. We have found a set of inequalities that define $P_{\Gamma}$, but this set is not minimal. We just need the equations defining the facets of $P_{\Gamma}$. Any facet is defined by a single equation, so we need to find subgraphs for which the equality defines a facet.
\begin{lem}\label{codim}
For a 2-connected graph $\gamma$ with no self-loops, $SP_{\gamma}$ is $|E(\gamma)|-1$ dimensional. 
\end{lem}
\begin{proof}
	Note that all integer points have the form $1_{S}$, where $S$ is complement of a spanning tree. Thus, all points of $SP_{\gamma}$ lie on the hyperplane where the sum of coefficients is $\ell_{\gamma}$. This implies that $SP_{\gamma}$ is at most $|E(\gamma)|-1$ dimensional. We show that $e_i-e_j$ can be constructed from differences of points of $SP_{\gamma}$.\\
	
	Let $e$ be an edge of $\gamma$ and let $C$ be a loop in $\gamma$ which contain $e$. This exists, since $\gamma$ is 2-connected. Let $T$ be a spanning tree that contains all edges of $C$ except $e$. Note that such a spanning tree exists, since any tree can be extended to a spanning tree. Let $e'$ be another edge of $C$. We claim that $T\cup e \setminus e'$ is a spanning tree. The reason is that $T\cup e$ has only one loop $C$, and removing any edge from it makes it a tree. Let $S=T^c$ and $S'= (T\cup e \setminus e')^c$. Then $1_{S}- 1_{S'}$ is $1_e- 1_e'$, hence, for any two edges $e,\ e'$ in the same loop, $1_e - 1_{e'}$ can be computed as a difference of points in $SP_{\gamma}$. Since $\gamma$ is 2-connected and does not have self loops, we can compute $1_e-1_{e'}$ for any pair of edges, and we find that $SP_{\gamma}$ is of  codimension one.
\end{proof}
For a self loop $e$ in $\Gamma$, $SP_e$ is just a point, which indeed is $|E(e)|-1=0$ dimensional. 
Note that $P_{\Gamma}$ is always $n$ dimensional, since it is equal to 
$SP_{\Gamma}+ E_n$ where $E_n$ is $n$ dimensional. 
\begin{lem}\label{dim}
	For a subgraph $\gamma \subset\Gamma$, $P_{\Gamma}\cap \{\tilde{x}: \langle 1_{\gamma}, \tilde{x}\rangle = \ell_{\gamma}\}$ is $SP_{\gamma}\times P_{\Gamma // \gamma}$, where $\Gamma // \gamma$ is constructed from contracting connected components of $\gamma$ to points. 
\end{lem}
\begin{proof}
	It is enough to find extreme points of $P_{\Gamma}\cap \{\tilde{x}: \langle 1_{\gamma}, \tilde{x}\rangle = \ell_{\gamma}\}$, since the equation defines a face of $P_{\Gamma}$. The equation $\langle \tilde{x} , 1_{\gamma}\rangle = \ell_{\gamma}$ for an extreme point $(a_e: e\in \Gamma)$, implies the vector $(a_e: e \in \gamma)$ is equal to $1_{S}$, where $S$ is the complement of a spanning tree in $\gamma$. A monomial $t^T$ can be written as $t^S t^{T-S}$, where $S$ has coefficients in $\gamma$ and $T-S$ corresponds to edges in $\Gamma- \gamma$ that are the edges in the contracted graph. Note that $t^{T-S}$ always corresponds to a monomial for the first or second Symanzik polynomial of $\Gamma//\gamma$. On the other hand $S+S'$, where $S$ comes from a spanning tree in $\gamma$ and $S'$ comes from a monomial in $\Gamma // \gamma$, is a monomial in the first or second Symanzik polynomial of $\Gamma$. Thus, we can identify the corresponding face $P_{\Gamma}$ with $SP_{\gamma}\times P_{\Gamma//\gamma}$
\end{proof}
\begin{theorem}\label{facet}
	For a graph $\Gamma$, the polytope $P_{\Gamma}$ is given by the inequality
	$$ \langle 1_{\Gamma}, \tilde{x}\rangle\leq \ell_{\Gamma}+1$$
	and inequalities
	$$\langle \gamma, \tilde{x}\rangle \geq\ell_{\gamma}$$
	 indexed by 2-connected subgraphs without self-loops $\gamma \subset \Gamma$, as well as
	inequalities 
	$$\langle 1_e,\tilde{x}\rangle= x_e \geq 1$$
	indexed by self loops $e\in \Gamma$.
	Replacing any inequality corresponding to a subgraph or self-loop $\gamma$ with equality, defines a facet of $P_{\Gamma}$ that is equal to $SP_{\gamma}\times P_{\Gamma // \gamma}$.
\end{theorem}
\begin{proof}
	By Proposition \ref{polytopeineq}, we know that these equation define the polytope. We just need to find the ones that define facets, i.e. are codimension one. By Lemma \ref{dim} the codimension is equal to the  codimension of $SP_{\gamma}$. The latter is equal to one for self loops and for 2-connected subgraphs without self-loops, by Lemma \ref{codim}. On the other hand, if we have a subgraph which is not 2-connected, then the complement of a spanning tree has to have $\ell_{\gamma_i}$ many edges in $\gamma_i$, where $\gamma_i$ are 2-connected components of $\gamma$. Thus, $SP_{\gamma}$ is at least of codimension 2. For a subgraph that is not a self loop but that contains a self loop $e$, we have two linear equalities for points of $SP_{\gamma}$, i.e. $\langle 1_{\gamma}, \tilde{x}\rangle = \ell_{\gamma}$ and $x_e=1$. This makes $SP_{\gamma}$ at least of codimension 2. 
\end{proof}

For a Feynman diagram $\Gamma$, the vector $\alpha$ is equal to 
$$(D/2,1,\dots,1) \, . $$
To find the inequalities defining the interior of the cone $C_A$, we need to replace inequalities with strict inequalities. Applying these to a vector $\alpha$, we find the necessary and sufficient condition for convergence of $K(\alpha,P_a)$, namely 
$$D/2(\ell+1) > |E(\Gamma)|\, . $$
If we have self-loops, then the inequality 
$$D/2=D\ell/2 < |E|=1$$
is not satisfied and the integral diverges. For all 2-connected subgraphs 
without self-loops $\gamma$ we have 
$$ D\ell/2 < |E(\gamma)| \, . $$
Note that for single edges this equality is satisfied, hence it is enough to check this for 2-connected subgraphs that are not single edges, i.e. for the so called $1PI$ subgraphs. 
\\
\begin{lem}
	Let $\Sigma$ be the semigroup generated by $A$ in $\mathbb{Z}^{n+1}$. Then $\Sigma$ is saturated. 
\end{lem}
\begin{proof}
	Assume $a=(a_0, \tilde{a})\in \mathbb{Z}^{n+1}$ is an integer point in $C_{A}$. We want to find $(\alpha_1, \alpha_2,\dots, \alpha_k)\in A^k$ such that $a=\sum_i \alpha_i$. Assume 	$$\langle \vec{\Gamma} , a\rangle > 0 \, . $$
	We claim that there exists $i\in\{1,\dots,n\}$, such that $a-e_i$ is in $C_{A}$. Assume this is not the case. Then, for each $i$, one of the inequalities is not satisfied by $a-e_i$. Note that we have 
	$$ \langle \vec{\Gamma}-e_0,a-e_i \rangle = \langle \vec{\Gamma} - e_0 , a\rangle -1 \leq0 \,. $$
Thus, for each $i$, there exists a subgraph $\gamma_i$ such that
$$\langle \vec{\gamma}_i , a-e_i\rangle < 0\, .$$
Since $a$ has integer coefficients and $a$ is in $C_{A}$, we must have 
$$\langle \vec{\gamma}_i, a\rangle= 0,$$
which is equivalent to 
$$\langle 1_{\gamma_i}, \tilde{a}/ a_0\rangle = \ell_{\gamma_i} \, . $$
By Lemma \ref{mayer} we have the same equality for a union of $\gamma_i$'s, i.e.
$$\langle 1_{\Gamma}, \tilde{a}/a_0\rangle = \ell_{\Gamma}\,. $$
This contradicts our assumption. As a result, we can write $a$ as $b+c$, where $b$ is an integer point in $C_{A}$ that satisfies $\langle\vec{\Gamma} , b \rangle= 0 $ and $c$ has positive integer coordinates with $c_0 =0$. Note that 
$$0 \geq \langle \vec{\Gamma} -e_0, b+c\rangle = -b_0 +\sum_i c_i, $$ 
which implies  $b_0 \geq \sum_i c_i$. Assume we can write $b= \sum_i \alpha'_i$ with $\alpha'_i$ in $\Sigma$ satisfying $\langle \vec{\Gamma}, \alpha'_i\rangle=0$. Since $b_0\geq \sum_i c_i$, we can distribute $\sum_i c_i$ many $e_i$'s among the $\alpha'_i$ and define $\alpha_i= \alpha'_i + e_{j_i}$.
To finish the proof, we need to show that the semigroup generated by the points of $A$ on the facet $\langle \vec{\Gamma} , x\rangle = 0$ is saturated. This has been shown in \cite{white}. 
\end{proof}
\begin{cor}
  $\mathbb{P}_{\Sigma}$ is projectively normal, i.e. $k[\Sigma]$ is integrally closed. By \cite{Hoch} the toric ideal has the Cohen-Macaulay property. 
\end{cor}
\begin{theorem}
	For a graph $\Gamma$, with the properties that 
	$$D/2(\ell_{\Gamma}+1)>|E(\Gamma)| \quad \text{ and } \quad D\ell_{\gamma} /2 < |E(\gamma)|,$$
	and for all $1PI$ subgraphs $\gamma$, the amplitude  is equal to 
	$$K((D/2,1,\dots,1), P_a), $$
	up to multiplication by rational numbers and powers of $\pi$.
This is a period of the motive $H^n(U_v, U_v\cap D)$. Furthermore, for any graph $\Gamma$, the lowest coefficient of the $\epsilon$-expansion of the amplitude agrees with the lowest coefficient of the $\epsilon$-expansion of $K((D/2+\epsilon,1,\dots,1))$, up to multiplication by rational numbers and powers of $\pi$. This  can be computed as a linear combination
$$\sum p_i(P_a) K(\alpha_i, P_a), $$
where the $p_i$'s are polynomials in the $P_a$'s with rational coefficients and all the
$K(\alpha_i,P_a)$'s are convergent and are periods of the motive $H^n(U_v, U_v \cap D)$.
\label{mmm}
\end{theorem}
\begin{proof}
	Note that the value of the gamma function at positive integers is an integer. The first part then follows from the Relation (\ref{ka}) between $K(\alpha,P_a)$ and $\mathcal{A}(\Gamma,P_a)$. For the second part, note that, using the relations in Lemma \ref{relations}, we can replace a divergent $K$ with a linear combination of convergent integrals, with polynomial coefficients in $P_a$. The poles arise by diving by terms $\langle w , \alpha\rangle$, which have a rational residue in the variable $\epsilon$. Moreover, we have division by the gamma function, which has rational residue at negative integers. Note that all the resulting convergent integrals are periods of $H^n(U_v,U_v\cap D)$ by Theorem \ref{motive}.
\end{proof}
\begin{proof}[Proof of Theorem \ref{mm}]
We have
$$c_0\mathcal{A}(D/2):= I(0,D/2,P_a,\vec{0})=\frac{\Gamma(\alpha_0)}{\Gamma(-\langle \vec{\Gamma}-e_0, \alpha\rangle)}K(\alpha,P_a),
$$
where $\alpha=(D/2,1,\cdots,1)$. By Remark \ref{zero}, zeros of $K(\alpha,P_a)$ cancel poles of the gamma function in the numerator. By Theorem \ref{poles}, poles of $K(\alpha,P_a)$ appear on semi non-resonant $-\alpha$'s. To find the semi non-resonant locus, we need to find the $-\Sigma$ translates of the facets. By description of the facets in Theorem \ref{facet}, $-\alpha$ is semi non-resonant iff
$$\langle \vec{\Gamma}-e_0, \alpha\rangle \in \mathbb{Z}_{\geq 0}$$ 
or 
$$\langle \vec{\gamma} , \alpha \rangle \in \mathbb{Z}_{\leq 0}$$
for a 2-connected subgraph $\gamma$. However, the poles of the gamma function at negative integers cancel the poles coming from the first equation. The second equation for $\alpha= (D/2,1,\cdots,1)$ is equivalent to

$$-D/2 \ell_{\gamma} + |E(\gamma)| \in \mathbb{Z}_{\leq 0}.$$
Thus, the first part of the theorem follows. The second part is a special case of Theorem \ref{mmm}. 
\end{proof}

\begin{rem}
	To find the $\epsilon$ expansion of the integral, using the relations in Lemma \ref{relations}, we can replace the differential form
$$\frac{t^{\tilde{\alpha}}}{(\sum_a P_a t^a)^{\alpha_0+\epsilon}} \frac{dt_1}{t_1}\dots\frac{dt_n}{t_n}$$
with differential forms that have logarithmic poles along the boundary, since division by zero appears when we want to replace $\alpha$ on a facet with a linear combination of points that are not on the facet. The boundary components correspond to the products of subgraphs and quotient graphs. Based on this observation, we can relate the Connes-Kreimer renormalization to the study of limiting mixed Hodge structures of \cite{bloch}. We will address renormalization in another forthcoming paper. 
\end{rem}
\newpage

\bibliographystyle{alpha.bst}
\bibliography{bibtex.bib}
\end{document}